\documentclass[12pt]{amsart}

\usepackage{microtype}

\usepackage{amsmath}
\usepackage{hyperref}
\usepackage{amsmath}
\usepackage{amssymb}
\usepackage{tikz}
\usepackage{color}
\usepackage{mathpartir}

\usepackage[centertags]{amsmath}

\usepackage[mathscr]{euscript}

\usepackage{amsthm}

\usepackage{tikz}

\usepackage{amsmath}

\usepackage{amssymb}

\usepackage{empheq}

\usepackage{newlfont}

\usepackage{amsfonts}

\usepackage{empheq}

\usepackage{newlfont}

\usepackage{graphicx}

\usepackage{color}

\usepackage{epsfig}

\usepackage{bm}

\usepackage{prooftree}

\newcommand{\der}{\vdash}

\newcommand{\prov}[3]{\triangleright\ #1\vdash_{#2}#3}

\newcommand{\sem}[2]{[\![#1]\!]_{#2}} 

\title{The $n$-dimensional Propositional Calculus}

\newcommand{\e}{\mathsf e}

\usepackage{microtype}
\usepackage{amsmath}
\usepackage{amsmath}
\usepackage{amssymb}
\usepackage{tikz}
\usepackage{color}

\newtheorem{fact}{Fact}

\newcommand\TTTT{%
 \textsf{T\kern-0.15em\raisebox{-0.55ex}T\kern-0.15emT\kern-0.15em\raisebox{-0.55ex}2}%
}

\providecommand{\U}[1]{\protect\rule{.1in}{.1in}}
\newtheorem{theorem}{Theorem}[section]
\theoremstyle{plain}

\newtheorem{corollary}[theorem]{Corollary}

\newtheorem{definition}[theorem]{Definition}
\newtheorem{example}[theorem]{Example}

\newtheorem{lemma}[theorem]{Lemma}

\newtheorem{proposition}[theorem]{Proposition}

\numberwithin{equation}{section}

\title{The higher dimensional propositional calculus  
}
\author[A. Bucciarelli]{A. Bucciarelli}\address{A. Bucciarelli, P.-L.Curien, A. Salibra, CNRS and Universit\'e Paris Cit\'e}\email{buccia@irif.fr,\, curien@irif.fr,\, salibra@irif.fr}
\author[P.-L. Curien]{P.-L. Curien}
\author[A. Ledda]{A. Ledda}\address{A. Ledda, F. Paoli, Universit\`a di Cagliari}\email{antonio.ledda@unica.it,\, paoli@unica.it}
\author[F. Paoli]{F. Paoli}
\author[A. Salibra]{A. Salibra}
\date{\today}
\thanks{{\bf Corresponding author}: {\bf Antonio Bucciarelli}, {\ttfamily buccia@irif.fr}
}
\begin{document}

\maketitle
\begin{abstract}
  In recent research, some of the present authors introduced the concept of an $n$-dimensional Boolean algebra and its corresponding propositional logic  $n\mathrm{CL}$, generalising the Boolean propositional calculus to $n\geq 2$ perfectly symmetric truth values. This paper presents a sound and complete sequent calculus for $n\mathrm{CL}$, named $n\mathrm{LK}$. We provide two proofs of completeness: one syntactic  and one semantic. The former implies as a corollary that $n\mathrm{LK}$ enjoys the cut admissibility property. The latter relies on the generalisation  to the $n$-ary case of the classical proof based on the Lindenbaum algebra of formulas and Boolean ultrafilters. 

\end{abstract}

\section{Introduction}
A research programme developed over the last few years \cite{first, second, SLP18, BS19, SBLP, SBLP2} has attempted to single out the features of Boolean algebras that let them encompass virtually all the ``desirable" properties one can expect from a variety of algebras. 

In the first part of this introduction, we summarise the main concepts and results obtained in this research. A Church algebra is an algebra $\mathbf{A}$ with term definable operations $q$ (ternary) and $0,1$ (nullary), such that for every $a,b\in A$, $q( 1,a,b) =a$ and $q( 0,a,b) =b$. The ternary operation $q$ intends to formalise the ``if-then-else" construct, widely employed in computer science. An element $e$ of a Church algebra $\mathbf{A}$ is called $2$\emph{-central} if $\mathbf{A}$ can be decomposed as the product $\mathbf{A}/\theta ( e,0) \times \mathbf{A}/\theta ( e,1) $, where $\theta ( e,0) $ ($\theta ( e,1) $) is the smallest congruence on $\mathbf{A}$ that collapses $e$ and $0$ ($e$ and $1$). According to a well-known result in the elementary structure theory of Boolean algebras, all elements of a Boolean algebra are $2$-central. More generally, Church algebras, where every element is $2$-central, were called \emph{Boolean-like algebras} in \cite{first}, since the variety of all such algebras in the language $( q,0,1) $ is term-equivalent to the variety of Boolean algebras. In \cite{SBLP, SBLP2} and in the present paper, on the other hand, they are called \emph{Boolean algebras of dimension }$2$. 

This approach can be generalised to algebras $\mathbf{A}$ having $n$
term definable elements  $\e_1,\ldots,\e_n$ and one
$(n+1)$-ary term definable operation $q$ satisfying $q(\e_i,x_1,\ldots,x_n)=x_i$, for all $i=1,\ldots,n$. Accordingly,  the notion of a $2$-central element can be extended to that of an $n$-central element,
that induces a decomposition of $\mathbf{A}$ into $n$, rather than just $2$, factors. In \cite{SBLP, SBLP2}, naturally enough, these algebras were called \emph{Church algebras of dimension }$n$ ($n\mathrm{CH}$s).
Free $\mathcal{V}$-algebras (for $\mathcal{V}$ a variety), rings with unit, semimodules over semirings --- hence, in particular, Boolean vector spaces --- give rise to $n\mathrm{CH}$s, where $n$ is generally greater than $2$.  Church algebras of dimension $n$, all of whose elements are $n$-central, were given the name of \emph{Boolean algebras of
dimension }$n$ ($n\mathrm{BA}$s). Since $n$-central elements are equationally definable, the $n\mathrm{BA}$s of a
specified type form a variety. Those $n\mathrm{BA}$s whose type includes just the $q$ operation and  the $n$ constants are called {\em pure}.
Varieties of  $n\mathrm{BA}$s happen to share many remarkable
properties with the variety of Boolean algebras. In particular:

\begin{itemize}
\item all subdirectly irreducible $n\mathrm{BA}$s of type $\tau$ have 
cardinality $n$; moreover, any $n\mathrm{BA}$ of type $\tau$ is a subdirect product of algebras of cardinality $n$;

\item any pure $n\mathrm{BA}$ is a subdirect power of the unique  pure $n\mathrm{BA}$ $\mathbf{n}$  of cardinality $n$, whose elements are considered as ``generalised truth values'';

\item for every $n\geq 2$ and type $\tau$, all $n\mathrm{BA}$s of type $\tau$ having  cardinality $n$ are primal. Moreover, every variety generated by an $n$-element primal algebra is a
variety of $n\mathrm{BA}$s.
\end{itemize}

In \cite{SBLP2} the theory of $n$-central elements was put to good use to yield an extension to
arbitrary semirings of the technique of \emph{Boolean powers}. We defined the
semiring power $\mathbf{A}[\mathbf{R}] $ of an algebra $\mathbf{A}$ by a semiring $\mathbf{R}$, and showed that any pure $n\mathrm{BA}$ $\mathbf{A}$ is a retract of the semiring power
$\mathbf{A}[ \mathbf{B}_{\mathbf{A}}] $ of $\mathbf{A}$ by what we called the inner Boolean algebra $\mathbf{B}_{\mathbf{A}}$
of $\mathbf{A}$. Foster's celebrated theorem on
primal algebras follows as a corollary from this result.

In \cite{BS19}  the connection between a noncommutative version of Boolean algebras, called skew Boolean algebras \cite{L89}, and pure $n\mathrm{BA}$s is explored. Moreover, the notion of a multideal, which plays an analogous role as filters and ideals do for Boolean algebras, is introduced and studied.

In \cite{SBLP} we focused on an application to logic. Just like Boolean algebras are the algebraic counterpart of classical propositional logic $\mathrm{CL}$, for every $n\geq 2$ we defined a logic $n\mathrm{CL}$ whose algebraic counterpart are $n\mathrm{BA}$s. We also proved that every tabular logic with a single designated value is a sublogic of some $n\mathrm{CL}$. Although we provided Hilbert-style calculi for each $n\mathrm{CL}$, the proof theory of these logics was not investigated in detail. 

In the present paper, we intend to develop more perspicuous calculi for the logics $n\mathrm{CL}$s, which may afford better insights into the behaviour of higher-dimensional logical connectives. The proof-theoretic framework of Gentzen's \emph{sequent calculi} appears as a promising candidate to achieve such a goal. 
While it is relatively easy to transform the sequent calculus for classical logic into a calculus for the equivalent logic $2\mathrm{CL}$, logics with a dimension greater than $2$ present trickier challenges.

To face these issues, we introduce the framework of \emph{higher-dimensional calculi}. Although our approach appears to be new, it presents similarities with other proof-theoretic methods available in the literature. In particular:
\begin{itemize}
    \item A successful generalisation of Gentzen's sequent calculi to finite-valued logics, dating back to Rousseau \cite{Rousseau}, consists in replacing $2$-sided sequents $\Gamma \Rightarrow \Delta$ by $n$-sided sequents $\Gamma_1 \mid ... \mid \Gamma_n$, one for each different truth value. The rules for each connective can be directly read off the truth table of the connective itself, which dictates what side of the conclusion-sequent should host the principal formula depending on the whereabouts of the auxiliary formulas in the premiss-sequents. Clearly, this process can be entirely automated: the programme Multlog (https://www.logic.at/multlog/) generates an $n$-sided calculus for any finite-valued logic whose truth tables are given as input \cite{BFZ}. Here, instead, we keep Gentzen's $2$-sided sequents but we split the turnstile $\vdash$ into $n$ different turnstiles $\vdash_i$ ($1 \leq i \leq n$), one for each truth value $\e_i$ of the calculus.
    \item Ordinary sequent calculi for a logic $\mathrm{L}$ aim at generating the $\mathrm{L}$-valid sequents and, in particular, the tautologies of $\mathrm{L}$; \emph{refutation systems} (see e.g. \cite{GPS} for a survey) aim at generating the $\mathrm{L}$-antivalid sequents and, in particular, the contradictions of $\mathrm{L}$. In the literature, a number of hybrid deduction-refutation systems have been proposed \cite{Goranko, Negri, Gore}, characterised by the presence of two different turnstiles $\vdash$ and $\dashv$ for deduction and refutation, respectively. The rules of these calculi, generally speaking, admit the simultaneous presence, as premisses or conclusions, of sequents formed with both turnstiles. In our calculus this idea is further generalised in so far as each truth value has its own associated turnstile.
    \item Hypersequents \cite{Avr} are finite multisets, or sequences, of ordinary sequents, each of which is said to be a \emph{component} of the hypersequent. While a sequent can be intuitively read as a claim to the effect that some formula in the consequent can be derived from the formulas in the antecedent, a hypersequent can be viewed as a disjunction of such derivability claims. Hypersequents are employed to obtain analytic calculi for several fuzzy, intermediate, relevant, and modal logics, for example (but not only) logics whose characteristic axioms have a disjunction as principal connective. \emph{Relational hypersequents} (see e.g. \cite{BaFer, CFM, CM}) generalise hypersequents in so far as different types of turnstiles can be used in different components. On some versions of the idea, components are expressions built by applying to formulas certain predicates of a first-order meta-language. It would be interesting to probe the affinities between this approach and the present one.
\end{itemize}

In this paper, for all $n\geq 2$, we define a sequent calculus $n\mathrm{LK}$ for the logic $n\mathrm{CL}$,
%
and we prove that $n\mathrm{LK}$ is sound and complete with respect to the natural $n$-valued semantics of $n\mathrm{LK}$ formulas. Soundness and completeness hold for all dimensions: we prove that the sequent $\Gamma\vdash_i\Delta$ is provable if and only if it is  valid, meaning that whenever an environment assigns the truth value $i$ to all the formulas in $\Gamma$, then it assigns the truth value $i$ to at least one formula in $\Delta$.
Our first proof of completeness is based on fact that the Lindenbaum algebra of $n\mathrm{LK}$ is an $n$BA. 
This semantic  proof relies on the cut rule of $n\mathrm{LK}$, which is needed in particular to show that equiprovability in all dimensions is an equivalence relation. 
We also provide another proof of completeness, based on syntactic properties of $n\mathrm{LK}$, in particular on the invertibility of some of the rules of the sequent calculus. This syntactic completeness result states in particular
that, whenever  a sequent  $\Gamma\vdash_i\Delta$   is  valid, then it is provable without using the cut rule,  thus ensuring that $n\mathrm{LK}$ enjoys the cut admissibility property.

\textbf{Outline of the paper}. In Section \ref{tutankamen} we recapitulate a few notions to make the present work reasonably self-contained. In particular, we give the definitons of Church algebras and of Boolean algebras of finite dimension, summarising the main results thus far obtained about them. In Section \ref{cleopatra} we introduce the sequent calculus $n\mathrm{LK}$. Section \ref{nefertiti} and its companion Section \ref{sec:class_cont} presents the case $n=2$.
 In section \ref{sec:closer} we study some syntatctic properties of $n\mathrm{LK}$, and we provide examples of higher dimensional proofs. Section \ref{sec:sem} presents the semantics of
$n\mathrm{LK}$, and the proof of soundness of the calculus. 
Section \ref{anubi} contains the semantic proof of completeness.
Some work is required, in particular, to show that the set of formulas of $n\mathrm{LK}$, quotiented by equiprovability, is the universe of a Boolean algebra of dimension $n$.
In Section \ref{sec:equiv} we prove that $n\mathrm{LK}$ is sound and complete for
$n\mathrm{CL}$, the propositional logic of the $n$BAs. 
As a corollary,  we have that  $n\mathrm{LK}$ is equivalent to the Hilbert-style axiomatisation of $n\mathrm{CL}$ presented in \cite{SBLP}. 
In Section \ref{sec:main} we provide a syntactic proof of completeness, 
 yielding the cut admissibility property.

\section{Preliminaries}\label{tutankamen}

 The proof of completeness of the  $n$-dimensional propositional calculus presented is Section \ref{anubi}
 relies on the fact that its Lindenbaum algebra, whose universe is the set of formulas quotiented by equiprovability, is  a
Boolean algebra of dimension $n$.
In this preliminary section we recall the main  definitions and facts  about  $n$BAs, making this paper reasonably self-contained.
For a start, let us introduce the concept of an $n$-dimensional Church algebra, from which the notion of an $n$-dimensional Boolean algebra has emerged.
The proofs of the results stated in this section can be found in 
\cite{SBLP}, \cite{SBLP2} and  \cite{BS19}.

\subsection{Church algebras of finite dimension\label{dobbiaco}}

In \cite{SBLP} we introduced \emph{Church algebras of dimension }$n$, algebras having $n$ designated elements $\e_1,\dots,\e_n$ ($n\geq 2$) and an operation $q$ of arity $n+1$
satisfying $q(\e_i,x_1,\dots,x_n)=x_i$.  The operator $q$ induces, through the so-called $n$-central elements, a decomposition of the algebra into $n$ factors.

\begin{definition}
\label{def:nCH}An algebra $\mathbf{A}$ of type $\tau $ is a \emph{Church algebra of
dimension} $n$ (an $n\mathrm{CH}$, for short) if there are term definable
elements $\e_{1}^{\mathbf{A}},\e_{2}^{\mathbf{A}},\dots ,\e_{n}^{\mathbf{A}}\in
A$ and a term operation $q^{\mathbf{A}}$ of arity $n+1$ such that, for all $%
b_{1},\dots ,b_{n}\in A$ and $1\leq i\leq n$, $q^{\mathbf{A}}(\e_{i}^{\mathbf{%
A}},b_{1},\dots ,b_{n})=b_{i}$. A variety $\mathcal{V}$ of type $\tau $ is a 
\emph{variety of algebras of dimension }$n$ if every member of $\mathcal{V}$
is an $n\mathrm{CH}$ with respect to the same terms $q,\e_{1},\dots ,\e_{n}$.
\end{definition}

{If $\mathbf{A}$ is an }$n\mathrm{CH}${,} then $\mathbf{A}_{0}=(A,q^{\mathbf{A}},\e^{\mathbf{A}}_{1},\dots ,\e_{n}^{\mathbf{A}})$ is the \emph{pure reduct} of $\mathbf{A}$.

Church algebras of dimension $2$ were introduced as Church algebras in \cite{MS08} and studied in \cite{first}.  Examples of Church algebras of dimension $2$ are Boolean algebras (with $q(x,y,z) =(x\wedge z)\vee (\lnot x\wedge y)$) or rings with unit (with $q( x,y,z) =xz+y-xy$). Next, we provide one example of Church algebras having dimension greater than $2${.} 
\begin{example}
\label{exa:partition} (\emph{$n$-Subsets}) Let $X$ be a set.
An \emph{$n$-subset} of $X$ is a sequence $(Y_{1},\dots ,Y_{n})$ of subsets 
 of $X$. We denote by $\mathrm{Set}_{n}(X)$ the family of all $n$-subsets of $X$.  $\mathrm{Set}_{n}(X)$ can be
 viewed as the universe of a Church algebra of dimension $n$, where
$\e_1=(X,\emptyset,\ldots,\emptyset), \e_2=(\emptyset,X,\emptyset,\ldots,\emptyset),\ldots,\e_n=(\emptyset,\ldots,\emptyset,X)$ and
 for all $\mathbf{y}^{0},\ldots,\mathbf{y}^{i}= (Y^i_{1},\dots ,Y^i_{n}),
 \ldots ,\mathbf{y}^{n}$:
$$
q( \mathbf{y}^0,\mathbf{y}^1,\dots ,\mathbf{y}^n) =(\bigcup\limits_{i=1}^{n}Y^0_{i}\cap Y_{1}^{i},\dots,\bigcup\limits_{i=1}^{n}Y^0_{i}\cap Y_{n}^{i}).
$$
\end{example}

{In \cite{vaggione}, Vaggione introduced the notion
of a \emph{central element} to study algebras whose complementary factor
congruences can be replaced by certain elements of their universes. If a
neat description of such elements is available, one usually gets important
insights into the structure theories of the algebras at issue. To list a few
examples, central elements coincide with central idempotents in rings with
unit, with complemented elements in $FL_{ew}$-algebras, which form the equivalent algebraic semantics of the full Lambek calculus with exchange and weakening, and with members of
the centre in ortholattices. In \cite{first}, T. Kowalski and three of the
present authors investigated central elements in Church algebras of dimension }$2$. 
In \cite{SBLP}, the idea was generalised to Church algebras of arbitrary finite dimension. 

\begin{definition}
\label{def:ncentral} If $\mathbf{A}$ is an $n\mathrm{CH}$, then $c\in A$
is called \emph{$n$-central} if the sequence of congruences $(\theta (c,\e_{1}),\dots
,\theta (c,\e_{n}))$ is an $n$-tuple of complementary factor congruences of $\mathbf{A}$. 
\end{definition}

The following characterisation of $n$-central elements, as well as the subsequent elementary result about them, were also proven in \cite{SBLP}.

\begin{theorem}
\label{thm:centrale} If $\mathbf{A}$ is an $n\mathrm{CH}$ of type $\tau $
and $c\in A$, then the following conditions are equivalent:

\begin{enumerate}
\item $c$ is $n$-central;

\item $\bigcap_{i\leq n}\theta (c,\e_{i})=\Delta $;

\item for all $a_{1},\dots ,a_{n}\in A$, $q(c,a_{1},\dots ,a_{n})$ is the
unique element such that $a_{i}\ \theta (c,\e_{i})\ q(c,a_{1},\dots ,a_{n})$,
for all $1\leq i\leq n$;

\item The following conditions are satisfied:
\begin{description}
\item[B1] $q(c,\e_{1},\dots ,\e_{n})=c$.

\item[B2] $q(c,x,x,\dots ,x)=x$ for every $x\in A$.

\item[B3] If $\sigma \in \tau $ has arity $k$ and $\mathbf x$ is an $n\times k$ matrix  of elements of $A$, of rows $\mathbf x_{1},\ldots,
\mathbf x_{n}$ and columns $\mathbf x^{1},\ldots,
\mathbf x^{k}$
  then
$
q(c,\sigma (\mathbf x_{1}),\dots ,\sigma (\mathbf x_{n}))=\sigma (q(c,\mathbf x^1),\dots ,q(c,\mathbf x^k)).$
\end{description}
\end{enumerate}
\end{theorem}

For any  $n$-central element $c$ and any $n\times n$ matrix $\mathbf x$  of elements of $A$, a direct consequence of (B1)-(B3) gives
\begin{description}
\item[B4]
$q(c,q(c,\mathbf x_{1}),\dots, q(c,\mathbf x_{n})) =q(c, x^1_1,x^2_{2}, \dots, x^n_{n})$.
\end{description}


\begin{proposition}\label{prop-closure} Let $\mathbf{A}$ be an $n\mathrm{CH}$. Then the set of all $n$-central elements of $\mathbf{A}$ is a subalgebra of the pure reduct of $\mathbf{A}$.
\end{proposition}

Hereafter, we denote by $\mathbf{Ce}_{n}(\mathbf{A})$ the algebra 
$(\mathrm{Ce}_{n}(\mathbf{A}),q,\e_{1},\dots ,\e_{n})$ 
of all $n$-central elements of an 
$n\mathrm{CH}$ $\mathbf{A}$.

\bigskip

\subsection{Boolean algebras of finite dimension\label{trallallero}}

Boolean algebras are Church algebras of dimension $2$ all of whose elements are 
$2$-central. It turns out that, among the $n$-dimensional Church algebras, those
algebras all of whose elements are $n$-central inherit many of the
remarkable properties that distinguish Boolean algebras. 

\begin{definition}
\label{mezzucci}An $n\mathrm{CH}$ $\mathbf{A}$ is called a \emph{Boolean
algebra of dimension $n$} ($n\mathrm{BA}$, for short) if every element of $A$
is $n$-central.
\end{definition}

By Proposition \ref{prop-closure}, the algebra $\mathbf{Ce}(\mathbf{A})$ of all $n$-central
elements of an $n\mathrm{CH}$ $\mathbf{A}$ is a canonical example of  $n\mathrm{BA}$. The class of all $n\mathrm{BA}$s of type $\tau $ is a variety of $n$CHs axiomatised by the identities
B1-B3 in Theorem \ref{thm:centrale}.

The following examples present two paradigmatic $n\mathrm{BA}$s, playing the 
role of the Boolean algebras of universe $\{0,1\}$ and $2^X$ for some set $X$,   respectively,
in the $n$-ary case.

\begin{example}\label{exa:n}
  The algebra $\mathbf{n}$ of universe $\{1,\ldots,n\}$ in the type of pure
 $n\mathrm{BA}$s such that $\e_i^\mathbf{n}=i$ and $q^{\mathbf{n}}(i,x_{1},\dots,x_{n}) =x_{i}$ 
 for every $i\leq n$ is a pure $n\mathrm{BA}$.
\end{example}

\begin{example}\label{exa:parapa} (\emph{$n$-Partitions}: Example \ref{exa:partition} continued) Let $X$ be a set. An \emph{$n$-partition} of $X$ is an $n$-subset $(Y_{1},\ldots ,Y_{n})$ of $X$ such that $\bigcup_{i=1}^{n}Y_{i}=X$ and $Y_{i}\cap Y_{j}=\emptyset $ for all $i\neq j$.
  The set of $n$-partitions of $X$ is closed under the $q$-operator defined in  Example \ref{exa:partition}
  and constitutes  the algebra of all $n$-central
elements of the $n\mathrm{CH}$ $\mathrm{Set}_{n}(X)$ of all $n$-subsets of $X$. Hence it is a pure  $n\mathrm{BA}$.
Notice that the algebra  of $n$-partitions of $X$, denoted by $\mathrm{Par}_{n}(X)$, is isomorphic to the $n\mathrm{BA}$ $\mathbf{n}^X$. In particular, the set of 2-partitions of a set X is nothing but the powerset of X in disguise.
\end{example}

Several remarkable properties of Boolean algebras find some analogue in the structure theory of $n\mathrm{BA}$s. 

\begin{theorem}
\label{lem:subirr} 
\begin{enumerate}
    \item An $n\mathrm{BA}$ $\mathbf{A}$ is subdirectly irreducible if and only if $|A|=n$.

    \item Every $n\mathrm{BA}$ $\mathbf{A}$ is isomorphic to a subdirect product of $\mathbf{B}_{1}^{I_{1}}\times \dots \times \mathbf{B}_{k}^{I_{k}}$ for some sets $I_{1},\dots ,I_{k}$ and some $n\mathrm{BA}$s $\mathbf{B}_{1},\dots ,\mathbf{B}_{k}$ of cardinality $n$.
    \item Every pure $n\mathrm{BA}$ $\mathbf{A}$ is isomorphic to a subdirect power of $\mathbf{n}^{I}$, for some set $I$.
\end{enumerate}
\end{theorem}

A subalgebra of the $n\mathrm{BA}$ $\mathrm{Par}_{n}(X)$ of the $n$-partitions on a set $X$, defined in Example \ref{exa:parapa},  is called a \emph{field of $n$-partitions on $X$}. The Stone representation theorem for $n\mathrm{BA}$s follows.

\begin{corollary} \label{cor:field-partitions}
Any pure $\mathrm{nBA}$ is isomorphic to a field of $n$-partitions on a suitable set $X$.
\end{corollary}

\subsection{Multideals and ultramultideals}
In this section we recall from \cite{BS19} some definitions and facts about the notion of {\em multideal}, whose role in the present context is analogous to the one played by filters and ideals in Boolean algebras. As a matter of fact, Proposition \ref{prop:ext1}, needed in the proof of completeness of the
$n\mathrm{LK}$, is new. 

\begin{definition}\label{def:fide}
Let $\mathbf A$ be an $n\mathrm{BA}$. A \emph{multideal} is an $n$-subset $(I_1,\dots, I_n)$  of $A$ such that:
\begin{enumerate}
\item[(m1)] $\e_k\in I_k$, for $k=1,\ldots,n$;
\item[(m2)] If $a\in I_j$, $b\in I_k$ and $c_1,\dots,c_n\in A$ then  $q(a,c_1,\dots,c_{j-1},b,c_{j+1},\dots,c_n)\in I_k$, for all $1\leq j,k\leq n$;
\item[(m3)]  If $a\in A$ and $c_1,\dots,c_n\in I_k$ then $q(a,c_1,\dots, c_n)\in I_k$,
 for all $1\leq k\leq n$.
\end{enumerate}
The set $I=I_1\cup\ldots\cup I_n$ is called the {\em carrier} of the multideal.
An \emph{ultramultideal} of $\mathbf A$ is a multideal whose carrier is $A$.
\end{definition}

The multideal $(I_1,\dots, I_n)$ is {\em proper} if it is a $n$-partition
of its carrier. 
The unique multideal of $\mathbf A$ that is not proper is the tuple $(A,\ldots,A)$.

As a matter of notation, if $\mathbf A$ is  an $n\mathrm{BA}$ and $x,y,z\in A$ 
let us write $t_i(x,y,z)$ for $q(x,y/\overline i,z/i)=q(x,y,\ldots,y,z,y,\ldots y)$, $z$ being the $(i+1)$-th argument of $q$.

\begin{definition}
  Let $\mathbf A$ be an $n\mathrm{BA}$ and $1\leq i\neq j\leq n$.
  \begin{itemize}
    \item
  The {\em Boolean centre of $\mathbf A$ with respect to $i,j$}, denoted by $\mathbf{A}_{ij}$, is
  the Boolean algebra of $2$-central elements of the $2\mathrm{CH}$
  $(A, t_i , \e_i , \e_j )$.
  \item The {\em $(i,j)$-coordinates of $a\in A$} are the elements $a_{(k)} = t_k (a, \e_i , \e_j )\in A_{ij}$, for $1\leq k\leq n$.
   \end{itemize}
\end{definition}

The following lemma relates multideals of $\mathbf A$ to ideals/filters of
$\mathbf{A}_{ij}$.

\begin{lemma}\label{lem:l1}
  Let $(I_1,\ldots,I_n)$  be a proper multideal of $\mathbf A$ and $1\leq i\neq j\leq n$. 
  Then $I_\star = A_{ij} \cap  I_i$  is a Boolean ideal of $\mathbf{A}_{ij}$ and
$I^\star = A_{ij} \cap I_j$  is the Boolean filter complement of $I_\star$ .
\end{lemma}

The following lemma characterises multideals in terms of coordinates.

\begin{lemma}\label{lem:l2}
  Let $(I_1,\ldots, I_n)$ be a proper multideal of a $n\mathrm{BA}$ $\mathbf A$,
  $1\leq i\neq j\leq n$  and  $b\in A$. Then we have, for all $1\leq r\leq n$,  $b\in I_r$  if and only if the $(i,j)$-coordinate $b_{(r)}$ of $b$ belongs to $I_j$.
\end{lemma}

The main lemma toward Proposition \ref{prop:ext1} is the following one.
\begin{lemma}\label{lem:l3}
Let $(I_1 ,\ldots, I_n )$  be a proper multideal, $1\leq i\neq j\leq n$  and $U$ be a Boolean ultrafilter of $\mathbf{A}_{ij}$  that
extends $I^\star = A_{ij} \cap  I_j$. 
Then we have:
\begin{itemize}
  \item[(i)] For all $a\in A$, there exists a unique $k$  such that $a_{(k)}\in U$.
  \item[(ii)] For $1\leq k \leq n$ let  $G_k=\{a\in A: a_{(k)}\in U\}$. Then
    $(G_1,\ldots,G_n)$ is a proper ultramultideal which extends  $(I_1,\ldots,I_n)$.
\end{itemize}
\end{lemma}

The proof of the following proposition relies on the corresponding result for Boolean algebras and on the connections between multideals and ideals/filters established above.

\begin{proposition}\label{prop:ext1}
  Let $(I_1 , \ldots, I_n )$ be a multideal of $\mathbf A$, $1\leq k\leq n$ and
  $a\in A$ be such that $a\not\in I_k$. Then there exists a proper ultramultideal $(G_1,\ldots,G_n)$ extending $(I_1 , \ldots, I_n )$ such that $a\not\in G_k$.

\end{proposition}
\begin{proof}
  Let us choose $1\leq i\neq j\leq n$. Since
  $a\not\in I_k$, then  by Lemma \ref{lem:l2}   $a_{(k)}\notin I_j$. Since the Boolean filter  $I^\star=A_{ij}\cap I_j$ is included in $I_j$, then $a_{(k)}\not\in I^\star$.
  Let $U$ be an ultrafilter of $\mathbf{A}_{ij}$ containing $I^\star$ such that $a_{(k)} \notin U$\footnote{For the existence of such an ultrafilter, see for instance Theorem 3.17  in Chapter IV of \cite{BS}.}.
By Lemma \ref{lem:l3}(ii), the sequence $G_r = \{b\in A : b_{(r)} \in U\}$ ($1\leq r\leq n$) is a proper ultramultideal extending $(I_1 , \ldots, I_n )$ such that $a\not\in G_k$.
\end{proof}

\section{The  $n$-dimensional Propositional Calculus }\label{cleopatra}

In this section we introduce the $n$-dimensional propositional calculus,
$n\mathrm{LK}$ in what follows, whose deduction rules are given in Figure \ref{nLK2}.

The formulas of $n\mathrm{LK}$ are built up starting from $n$ different constants $\e_1,\ldots,\e_n$, denoting the generalised truth values $1,\ldots, n$.  Hence the denomination ``higher-dimensional calculus''.

A peculiar characteristic of $n\mathrm{LK}$ is that it has a unique connective, $q$,  of arity $n+1$. The intended meaning of this connective is the following:
the truth value of the formula $q(F,G_1,\ldots,G_n)$ is that  of $G_k$,
$k$ being the truth value of $F$. When $n=2$, this is the usual interpretation of 
$\mathrm{if\_then\_else(F,G_1,G_2)}$,  by renaming the truth values: $\mathrm{true}=1$,
$\mathrm{false}=2$.

Another distinctive   feature of this deductive system is 
that each dimension $i\in\hat n=\{1,2,\ldots,n\}$ has its own turnstile $\der_i$. 
In the $2$-dimensional case,
this gives rise to the turnstiles 
 $\der_1$ and $\der_2$, the former deriving tautologies, and the 
latter deriving contradictions.
Entailments in the various  dimensions are symmetric, in the sense expressed by the rule
$(\mathtt{Sym})$ of Figure \ref{nLK2}, that can be instantiated as follows, for $n=2$ and for a given propositional formula $F$:
$$
\infer
{\der_1\ F}
{\der_2 F^{(12)}}
$$
where $F^{(12)}$ 
is obtained by switching $\e_1$ and $\e_2$ in $F$ (see Definition \ref{def:perm}).
  Intuitively $F^{(12)}\equiv q(F,\e_2,\e_1)$, whereas
  $F\equiv q(F,\e_1,\e_2)$. Notice 
 that in the instantiation of the rule $(\mathtt{Sym})$ above we have used the fact that
  $(F^{(12)})^{(12)}=F$, shown in Lemma \ref{Lemma:involution}.
  .
  
In the case $n=2$ there is a unique way of switching the constants $\e_1$ and $\e_2$, corresponding to the classical negation, 
whereas in general there are $\binom n 2$ possible ``negations'', and more generally
there are $n!-1$ ways of perturbing a formula by permuting the constants $\e_1,\ldots,\e_n$ in it.
The permutations, and in particular the exchanges, are primitive in our syntax for
the formulas of $n\mathrm{LK}$.

\subsection{Syntax of  $n\mathrm{LK}$}

Let $S_n$ denote the group of permutations of $\hat n$, 
ranged over by $\pi,\rho,\sigma, \ldots$,
and  let $V=\{X,Y,Z,X',\ldots\}$ be a countable set of propositional 
variables.

\begin{definition}
A {\it formula} of  $n\mathrm{LK}$ is either:
\begin{itemize}
\item a  {\em decorated propositional variable} $X^\pi$, or 
\item one of the {\em constants}  $\e_1,\ldots,\e_n $, or 
\item a {\em compound formula} $q(F,G_1,\ldots,G_{n})$, where $F,G_1,\ldots,G_{n}$ are formulas.
\end{itemize}
We write $\mathcal F_n$ for the set of formulas.
\end{definition}

The choice of decorating propositional variables by permutations deserves 
some explanations. As a matter of fact,  
$X^\pi$ will be proven to be logically equivalent to $q(X,\e_{\pi(1)},\ldots,\e_{\pi(n)})$ (see Lemma \ref{Lemma:permutations}) and we could have stipulated
that atomic formulas are either  propositional variables or constants. Nevertheless,
the choice of decorating variables by permutations  eases the task of defining an involutive form of
negation. 

\begin{definition}\label{def:perm}
Given  $\rho\in S_n$ and a formula $F\in\mathcal F_n$, let $F^\rho$ be the formula 
inductively defined by:

$F^{\rho}=\left\{\begin{array}{ll}
X^{\rho\circ\pi}&\mbox{ if } F=X^\pi\\
\e_{\rho(k)}&\mbox{ if } F=\e_k\\
q(F,G_1^{\rho},\ldots,G_n^{\rho})&\mbox{ if } F=q(F,G_1,\ldots,G_n).
\end{array}
\right.
$
 \end{definition}

\begin{lemma}\label{Lemma:perm}
For each $F\in{\mathcal F}_n$, $\pi,\rho\in  S_n$:
$(F^{\pi})^\rho=F^{\rho\circ\pi}$ and $F^{id}=F$.
\end{lemma}
\begin{proof}
Easy induction on $F$.
\end{proof}

The simplest non-trivial permutations are the  {\em exchanges}, noted  $(ij)$.
Given $i,j\in\hat n$, the formula $F^{(ij)}$ is 
the negation of $F$  relatively to the dimensions $i$ and $j$.
This kind of negation is involutive.

\begin{lemma}\label{Lemma:involution}
For each $F\in{\mathcal F}_n$,  $i,j\in\hat n$, we have that
${F^{(ij)}}^{(ij)}=F$.
\end{lemma}
\begin{proof}
By Lemma \ref{Lemma:perm}, since $(ij)\circ (ij)=id$.
\end{proof}

Thus, in $n\mathrm{LK}$, the negations are {\em strongly} involutive in the sense that 
${F^{(ij)}}^{(ij)}$ and $F$ are {\em the same} formula.

Hence,  in particular, replacing ${F^{(ij)}}^{(ij)}$ by $F$ is not a semantic shortcut.
Those are nothing but two different ways of writing the same formula.
Throughout the paper, we will write $F$ for 
 ${F^{(ij)}}^{(ij)}$.

It is worth noticing that, by defining the size of a formula as
the maximal nesting level of compound formulas in it, $F$ and
$F^{(ij)}$ have the same size, for all formulas. 
In an inductive proof, when dealing 
with  $q(F,G_1,\ldots,G_n)$, the inductive  hypothesis may be applied to 
$F,G_1,\ldots,G_n$, and to $F^{(ij)},G_1^{(ij)},\ldots,G_n^{(ij)}$ as well.

{\em Contexts} ranged over by $\Gamma, \Delta,\ldots$ are finite multisets
of formulas, written as sequences. If $\Gamma=F_1,\ldots,F_n$, the notation 
$\Gamma^{(ij)}$ stands for $F_1^{(ij)},\ldots,F_n^{(ij)}$. Also, in some premises
of the deduction rules of Figure \ref{nLK2},
$\Big\{\Gamma_i\der_i \Delta_i\Big\}_{i\in I}$, $I\subseteq \hat n$, stands for a sequence
of $|I|$ premises, one for each $i\in I$.

The notation $$\prov \Gamma i \Delta$$ means that the sequent 
$\Gamma\der_i\Delta$ is provable, using the rules of Figure \ref{nLK2}.

\begin{figure}
\footnotesize{
\begin{center}
$\begin{array}{c}
\infer{\ }{ \der_i\e_i}\ \mathtt{Const}
\quad \quad 
\infer{\ \pi^{-1}(i)=\rho^{-1}(i) }{ X^\pi\der_i X^\rho}\ \mathtt{Id}
\quad \quad 
\infer
{\Gamma^{(ij)}\der_j\Delta^{(ij)}}
{\Gamma\der_i\Delta}\ \mathtt{Sym}

\\ \\

 \infer
{\Gamma\der_i F,\Delta\ \ i\not=j}
{\Gamma, F^{(ij)}\der_i\Delta}\ {\mathtt{NegL}}
\quad \quad 
\infer
{\Big\{\Gamma^{(ij)},F\der_j\Delta^{(ij)}\Big\}_{j\in\hat n\setminus\{i\}}}
{\Gamma\der_i F, \Delta} {\mathtt{NegR}}

\\ \\

\infer
{\Big\{\Gamma^{(ij)}, F,G_j^{(ij)}\der_j \Delta^{(ij)}
\Big\}_{j\in\hat n}}
{\Gamma, q(F,G_1,\ldots,G_{n})\der_i \Delta}\ {\mathtt{qL}}
\quad \quad \quad \quad 
\infer
{\Big\{\Gamma^{(ij)}, F\der_j G_j^{(ij)},\Delta^{(ij)}\Big\}_{j\in\hat n}
}
{\Gamma\der_i q(F,G_1,\ldots,G_n),\Delta}\ {\mathtt{qR}}

\\ \\
\infer
{\Gamma, F\der_i \Delta\ \ \   \Gamma\der_i F, \Delta}
{\Gamma\der_i \Delta}\ {\mathtt{Cut}}

\\ \\
\infer
{\Gamma\der_i\Delta }
{\Gamma,F \der_i \Delta}\ {\mathtt{WeakL}}\quad \quad 
\infer
{\Gamma\der_i\Delta}
{\Gamma\der_i F,\Delta}\ {\mathtt{WeakR}}

\\ \\
\infer
{\Gamma, F, F\der_i\Delta }
{\Gamma,F \der_i \Delta}\ {\mathtt{ConL}}\quad \quad 
\infer
{\Gamma\der_i F,F,\Delta}
{\Gamma\der_i F,\Delta}\ {\mathtt{ConR}}

\end{array}$
\caption{The system  $n\mathrm{LK}$. Indices
$i,j,k$ range over the set $\hat n=\{1,\ldots,n\}$  of dimensions, $\pi,\rho\in S_n$ and $X\in V$. $\Gamma,\Delta$ are finite multisets of formulas, represented as lists.
}
\label{nLK2}
\end{center}
}
\end{figure}

The deduction rules of the systems may be justified using the notion of $n$-partition presented in  Example \ref{exa:parapa}.

If $G^1=(G_1^1,\ldots,G^1_n),G^2,\ldots, G^r,F^1,\ldots F^s$ are $n$-partitions of a set $X$ and $1\leq i \leq n$ then, intuitively,  the sequent 
$G^1,\ldots, G^r \vdash_i  F^1,\ldots F^s$ states that

$$\bigcap_{l=1}^r G^l_i\subseteq \bigcup_{l=1}^s F^l_i
$$

This remark provides a handy guideline to verify the validity of the rules of figure \ref{nLK2}, by considering that:
\begin{itemize}

\item $\e_i$ stands for  the $n$-partition $(\emptyset,\ldots, \emptyset,X,\emptyset,\ldots, \emptyset)$, $X$ being at the $i$-th position.
\item If $Y=(Y_1,\ldots,Y_n)$ is an $n$-partition  and $\sigma$ is a permutation, then 
$Y^\sigma=(Y_{\sigma(1)},\ldots,Y_{\sigma(n)})$.
\item The operator $q_n$ acts on $n$-partitions as defined in Example \ref{exa:parapa}.
\end{itemize}

Let us look now at rule $(\mathtt{NegL})$, for instance: the premise states that the intersection
of the $i$-th components of the elements of $\Gamma$ is included in the union of the $i$-th component of
$F$ and of the $i$-th components of the elements of $\Delta$. Let us write this
as $\bigcap \Gamma_i \subseteq F_i\cup \bigcup \Delta_i$.
Now, $F_i$ and $F_j= F^{(ij)}_i$ are disjoint, 
since $i\not=j$. This ensures that 
$\bigcap \Gamma_i \cap F^{(ij)}_i\subseteq \bigcup \Delta_i$, which is  what is stated in the conclusion of the rule. This kind of argument applies to all the rules of $n\mathrm{LK}$.



\subsection{The classical case: syntax}\label{nefertiti}
Before going further, we compare the usual propositional calculus $\mathrm{PC}$ and  the
2-dimensional propositional calculus $2\mathrm{LK}$. 
Formulas of $\mathrm{PC}$ are built with the
constants 0,1 (written in the same way as the corresponding truth values), variables $X\in V$, and the connectives $\neg$, $\wedge$ and  $\vee$. The sequent calculus is given in Figure \ref{PC}.
It is sound and complete: a sequent $\Gamma \der \Delta$ is
provable iff it is valid. 
We have chosen a Ketonen-style formulation 
which is close
to the $2\mathrm{LK}$ sequent calculus.

\begin{figure}
\footnotesize{
\begin{center}
$\begin{array}{c}
\infer{\ }{ \der 1}\ (\mathtt{Const1})
   \quad
   \quad
\infer{\ }{0  \der }\ (\mathtt{Const0})

  \quad
   \quad
   
\infer{\ }{ P \der P}\ (\mathtt{Id})

\\ \\

 \infer
{\Gamma, P,Q \der \Delta}
{\Gamma, P\wedge Q \der \Delta}\ {(\mathtt{\wedge L})}
\quad \quad 
 \infer
{\Gamma\der P,\Delta \quad \Gamma\der Q,\Delta}
{\Gamma \der P\wedge Q, \Delta}\ {(\mathtt{\wedge R})}

\\ \\

  \infer
{\Gamma, P\der \Delta \quad \Gamma,Q\der \Delta}
{\Gamma, P\vee Q \der \Delta}\ {(\mathtt{\vee L})}
\quad \quad 

   \infer
{\Gamma \der P,Q,\Delta}
{\Gamma \der P\vee Q, \Delta}\ {(\mathtt{\vee R})}

\\ \\

  \infer
{\Gamma \der P, \Delta}
{\Gamma, \neg P \der  \Delta}\ {(\mathtt{\neg L})}

\quad \quad 
\infer
{\Gamma, P \der  \Delta}
{\Gamma  \der \neg P, \Delta}\ {(\mathtt{\neg R})}

\\ \\
\infer
{\Gamma, F\der \Delta\ \ \   \Gamma\der F, \Delta}
{\Gamma\der \Delta}\ {(\mathtt{Cut})}

\\ \\
\infer
{\Gamma\der \Delta }
{\Gamma,F \der  \Delta}\ {(\mathtt{WeakL})}\quad \quad 
\infer
{\Gamma\der \Delta}
{\Gamma\der  F,\Delta}\ {(\mathtt{WeakR})}

\\ \\
\infer
{\Gamma, F, F\der \Delta }
{\Gamma,F \der \Delta}\ {(\mathtt{ConL})}\quad \quad 
\infer
{\Gamma\der F,F,\Delta}
{\Gamma\der F,\Delta}\ {(\mathtt{ConR})}

\end{array}$
\caption{The Propositional Calculus.}
\label{PC}
\end{center}
}
\end{figure}

In order to compare $\mathrm{PC}$ and $2\mathrm{LK}$, we provide in Figure \ref{Transl}
translations from $\mathrm{PC}$ formulas to $2\mathrm{LK}$ formulas and conversely.
As a matter of notation, $P,Q,R$ denote formulas of $\mathrm{PC}$, and $F,G,H$ formulas of $2\mathrm{LK}$, and we write $\Gamma^\circ$ and $\Gamma^\bullet$ to denote the translations of sequences of formulas.
Notice that, for $n=2$ there exist two permutations of $\hat n$: the identity $id$ and the exchange $(12)$. 

\begin{figure}
\footnotesize{
\begin{center}
  \begin{tabular}{l|l}
    \begin{tabular}{l}
 $0^\circ=\e_2\ \ \ $ \\
 $1^\circ=\e_1\ \ \ $\\
      $X^\circ=X^{id}\ \ \ $\\
  $(\neg P)^\circ= (P^\circ)^{(12)}\ \ \ $\\
      $(P\wedge Q)^\circ= q(P^\circ,Q^\circ,\e_2)\ \ \ $ \\
     $(P\vee Q)^\circ= q(P^\circ,\e_1,Q^\circ)\ \ \ $
 \end{tabular}
      &
       \begin{tabular}{l}  
$\ \ \ \e_2^\bullet=0$ \\
 $\ \ \ \e_1^\bullet=1$ \\ 
         $\ \ \  (X^{id})^\bullet=X$ \\
 $\ \ \  (X^{(12)})^\bullet=\neg X$\\
         
 $\ \ \  q(F,G,H)^\bullet= (F^\bullet \wedge G^\bullet)\vee (\neg F^\bullet \wedge H^\bullet) $\\

       \end{tabular}

  \end{tabular}
\caption{The translations $2\mathrm{LK}$${\leftrightarrow}$$\mathrm{PC}$.}
\label{Transl}
\end{center}
}
\end{figure}


In Section  \ref{sec:class_cont} we will show that $\mathrm{PC}$ and $2\mathrm{LK}$ are equivalent in the
sense that:
\begin{enumerate}
\item Given two sequences of $\mathrm{PC}$ formulas $\Gamma,\Delta$, we have $\prov \Gamma {} \Delta$ iff $\prov {\Gamma^\circ} 1 {\Delta^\circ}$;
  \item Given two sequences of $2\mathrm{LK}$  formulas $\Gamma,\Delta$, 
we have $\prov \Gamma {1} \Delta$ iff $\prov {\Gamma^\bullet} {} {\Delta^\bullet}$.
\end{enumerate}
Since both calculi are sound and complete with respect to  suitable notions
of validity of sequents, the easiest way to prove this equivalence is through
the semantics. 

We conclude this section by presenting two examples of simulations
of $\mathrm{PC}$ by $2\mathrm{LK}$ and vice-versa. 
When translating  sequents of the form  $\Gamma \vdash_2\Delta $ from $2\mathrm{LK}$ to $\mathrm{PC}$, we use the following property,
that is an easy consequence of the second item above, by using the rule $\mathtt{Sym}$ of $2\mathrm{LK}$:
\begin{equation}\label{fact:sym}
 \forall\  \Gamma,\Delta\text{ sequences of $2\mathrm{LK}$ formulas, we have }  \prov \Gamma 2 \Delta \text{ iff }\prov {{\Gamma^{(12)}}^\bullet} {} {{\Delta^{(12)}}^\bullet}
\end{equation}

\begin{example}
  \begin{itemize}
    \item
Consider the following instance of the  rule $\mathtt{\wedge R}$ of $\mathrm{PC}$:
$$ \infer
{\der X \quad \der Y}
{\der X\wedge Y}\ {(\mathtt{\wedge R})}
$$
where $X$ and $Y$ are propositional variables.
The  $2\mathrm{LK}$ formula  $(X\wedge Y)^\circ= q(X^{id}, Y^{id},\e_2)$
may be proven as follows, under the hypothesis $\der_1 X^{id}$ and  $\der_1 Y^{id}$
$$ \infer
{
  \infer {\der_1 Y^{id}} {X^{id}\der_1 Y^{id}}\ {(\mathtt{WeakL})}
  \quad
  \infer {\infer{\infer {\der_1X^{id} } {X^{(12)}\der_1} \ {(\mathtt{NegL})} }
{X^{id}\der_2}{(\mathtt{Sym})}
  }  {X^{id}\der_2 \e_1}\ {(\mathtt{WeakR})}
}
{\der_1 q(X^{id}, Y^{id},\e_2)}\ {(\mathtt{qR})}
$$

\item The left rule for $\wedge$ 
$$ \infer
{X,Y\der }
{X\wedge Y \der} \ {(\mathtt{\wedge L})}
$$
is translated as follows:

$$ \infer
{ \infer {\ }
  {
    X^{id}, Y^{id} \der_1
    }  \ {(\mathtt{Id})}
  \quad
  \infer {\infer {\infer {\infer{\ } {\der_1 \e_1}  \ {(\mathtt{Const})}  } {\e_2\der_1} \ {(\mathtt{NegL})}}{\e_1\der_2}\ {(\mathtt{Sym})}        }  {X^{id}, \e_1\der_2}\ {(\mathtt{WeakL})}
}
{   q(X^{id},Y^{id},\e_2 )  \der_1}\ {(\mathtt{qL})}
$$

\end{itemize}
\end{example}






\begin{example}

 \begin{itemize}

    \item
Consider the following  instance of the rule $\mathtt{qR}$ of $2\mathrm{LK}$:
$$ \infer
{X^{id}\der_1 Y ^{id} \quad    X^{id}\der_2 Z^{(12)} }
{\der_1 q(X^{id},Y^{id},Z^{id})}\ {(\mathtt{qR})}
$$
The $\mathrm{PC}$ formula  $q(X^{id},Y^{id},Z^{id})^\bullet=(X\wedge Y)\vee (\neg X \wedge Z)$
may be proven as follows, under the hypothesis $X\der Y$ and $\neg X\der Z$ (notice that the latter hypothesis is obtained by  $X^{id}\der_2 Z^{(12)}$, by  applying
\ref{fact:sym}):

{\small 
$$ 
  \infer {
    \infer {
      \infer { \infer {\infer{\ }{X\der X}\ {(\mathtt{Id})} \ \quad X\der Y} {X\der X\wedge Y}\ {(\mathtt{\wedge R})}     }{X \der (X\wedge Y)\vee (\neg X \wedge Z)}\ {(\mathtt{WeakR,\vee R})}}
           { \der \neg X, (X\wedge Y)\vee (\neg X \wedge Z)   }\ {(\mathtt{\neg R})}
    \quad
  \infer{  \infer {\infer{\ }{\neg X\der \neg X} \ {(\mathtt{Id})}  \ \quad \neg X\der Z} {\neg X\der \neg X \wedge Z}\ {(\mathtt{\wedge R})}   }{ \neg X\der (X\wedge Y)\vee (\neg X \wedge Z)  } \ {(\mathtt{WeakR,\vee R})}
  } {\der (X\wedge Y)\vee (\neg X \wedge Z) }\ {(\mathtt{Cut})}
$$
}

\item

  The rule ${(\mathtt{qL})}$:

  $$ \infer
{X^{id},Y^{id}\der_1 \quad    X^{id}, Z^{(12)}\der_2 }
{ q(X^{id},Y^{id},Z^{id})\der_1}\ {(\mathtt{qL})}
$$
is translated as follows (notice the use of  \ref{fact:sym} in the translation of the rightmost premise   $X^{id}, Z^{(12)}\der_2$) :

  $$ \infer
  {
    \infer {X, Y \der} {X\wedge Y\der    }\ {(\mathtt{\wedge L})}
  \quad
  \infer {\neg X, Z \der} {\neg X \wedge Z   \der    }\ {(\mathtt{\wedge L})}
}
{  (X\wedge Y)\vee (\neg X \wedge Z)\der}\ {(\mathtt{\vee L})}
$$

\end{itemize}
  
\end{example}


\section{A closer look at the syntax of  $n\mathrm{LK}$}\label{sec:closer}
In this section, ended by two examples of proofs in 3LK,  we study some syntactic properties of $n\mathrm{LK}$, and in particular we introduce two derivable rules that will be used in the semantic completeness result of Section \ref{anubi}.  


We show first that the property expressed in the axiom $\mathtt{Id}$ for decorated propositional variables
may be proved for all the formulas. This is proved in Corollary \ref{cor:id}.


\begin{lemma}\label{lem:neg0}
  Let $\sigma,\rho\in S_n$ and $i$ be such that $\rho^{-1}(i)\neq\sigma^{-1}(i)$. Then
     $\prov {F^\sigma,F^\rho} i {}$, for every  formula $F$.
\end{lemma}
\begin{proof}
  By induction on $F$.

(i)  If $F$ is a decorated variable $X^\pi$, let us choose $k=\rho(\sigma^{-1}(i))$.
  From the hypothesis $\rho^{-1}(i)\neq\sigma^{-1}(i)$ it follows immediately that $k\neq i$. Consider

  $$
  \infer
      {
        \infer
            {(\sigma\circ\pi)^{-1}(i) =  ((ik)\circ\rho\circ \pi)^{-1}(i)}
            {X^{\sigma\circ \pi}\der_i X^{(ik)\circ\rho\circ \pi}\ \ i\not=k}
            {\mathtt{Id}}
      }
      {X^{\sigma\circ \pi}, X^{(ik)\circ(ik)\circ\rho\circ \pi}\der_i}
      {\mathtt{NegL}}
  $$

  The axiom  $\mathtt{Id}$ can be applied since 
  $  ((ik)\circ\rho\circ \pi)^{-1}(i)=\pi^{-1}(\rho^{-1}((ik)^{-1}(i)))=
  \pi^{-1}(\rho^{-1}(k))= 
  \pi^{-1}(\rho^{-1}(\rho(\sigma^{-1}(i))))=\pi^{-1}(\sigma^{-1}(i))=(\sigma\circ\pi)^{-1}(i)$.
    The conclusion of this proof  is $F^\sigma,F^\rho\der_i$, as expected.

  (ii)  If $F=\e_j$ then either $\sigma(j)\neq i$ or $\rho(j)\neq i$. Let us suppose
  without loss of generality that $\sigma(j)\neq i$. Consider

  $$
  \infer
      {
        \infer
            {
              \infer
                  {\ }
                  {\der_i\e_{\sigma(j)}^{(i,\sigma(j))} \quad  i\neq \sigma(j)}
                  {\mathtt{Const}}
            }
            {\e_{\sigma(j)}\der_i}
            {\mathtt{NegL}}
      }
      {\e_{\sigma(j)},\e_{\rho(j)}\der_i}
      {\mathtt{WeakL}}
 $$
The  conclusion of this proof  is $F^\sigma,F^\rho\der_i$, as expected.

    (iii)  If $F=q(G,H_1,\ldots,H_n))$ then proceeding as follows:

  $$
  \infer {
        \infer
            {\left\{G^{(j,l)},H_j^{(j,l)\circ(i,j)\circ\rho},G, H_l^{(j,l)\circ(i,j)\circ\sigma}\der_l\right\}_{l\in\hat n}
            }
            {\left\{G,H_j^{(ij)\circ\rho},q(G,H_1^\sigma,\ldots,H_n^\sigma))^{(ij)}\vdash_j\right\}_{j\in\hat n}}
            {(\mathtt{qL})}}
            {q(G,H_1^\rho,\ldots,H_n^\rho)),q(G,H_1^\sigma,\ldots,H_n^\sigma))\vdash_i}
      {(\mathtt{qL})}
  $$

      we obtain $n^2$ branches, for $1\leq j,l\leq n$, whose leaves are
      of the form $$G^{(j,l)},H_j^{(j,l)\circ(i,j)\circ\rho},G, H_l^{(j,l)\circ(i,j)\circ\sigma}\der_l$$
      If $j=l$ then
 $H_j^{(j,l)\circ(i,j)\circ\rho},H_l^{(j,l)\circ(i,j)\circ\sigma}\der_l$ is provable by the inductive hypothesis, and hence the sequent above may be proved by using the rule
 ($\mathtt{WeakL}$), whereas if $j\neq l$ then $G^{(j,l)},G \der_l$,  is provable by the inductive hypothesis, and again the sequent above may be proved by weakening.
 This provides a proof of $F^\sigma,F^\rho\der_i$, as expected.

\end{proof}

Lemma \ref{lem:neg0} allows us to generalise the axiom $\mathtt{Id}$ to all the formulas.

\begin{corollary}\label{cor:id}
  For all formulas $F$, $1\leq i\leq n$ and permutations $\sigma,\rho$, if 
  $\sigma^{-1}(i)=\rho^{-1}(i)$ then $\prov {F^\sigma} i  {F^\rho}$.

\end{corollary}

\begin{proof}

$$\infer
{\Big\{F^{(ij)\circ\sigma},F^\rho\der_j\Big\}_{i\neq j}}
{F^\sigma \der_i F^\rho } {\mathtt{NegR}}
$$

Each of the $n-1$ premises 
is provable by Lemma \ref{lem:neg0} since $((ij)\circ\sigma)^{-1}(j)=\sigma^{-1}(i)=\rho^{-1}(i)\neq \rho^{-1}(j)$, for all $j\neq i$.
\end{proof}

Notice that Corollary \ref{cor:id} shows in particular that  $\prov F i F$
holds for every  formula $F$ and dimension $i$.
The next lemma introduces two derivable rules that we use in the rest of the paper. 

\begin{lemma}
  The rules 
   $$\infer
{\Gamma^{(ij)}\der_i F,\Delta^{(ij)}\ \ i\not=k}
{\Gamma, F^{(jk)}\der_j\Delta}\ {\mathtt{Neg1}}
\quad \quad 
 \infer
{\Gamma^{(ij)}\der_i F,\Delta^{(ij)}\ \ j\not=k}
{\Gamma, F^{(ik)}\der_j\Delta}\ {\mathtt{Neg2}}
$$
are derivable.
\end{lemma}
\begin{proof}
  Concerning $\mathtt{Neg1}$  we have: if $j=i$, then 
  $\mathtt{Neg1}$ is an instance of $\mathtt{NegL}$. Otherwise:
  $$
  \infer
      {
        \infer
        {
        \infer
            {
              \infer
                  { \Gamma^{(ij)} \der_i F, \Delta^{(ij)}\quad i\neq k}
                  {\Gamma^{(ij)},F^{(ik)}\der_i \Delta^{(ij)}}
                  {\mathtt{NegL}}
            }
            {\Gamma,F^{(ij)\circ(ik)}\vdash_j\Delta}
            {\mathtt{Sym}}
        }
        {\Gamma,F^{(ij)\circ(ik)},F^{(jk)}\vdash_j\Delta}
         {\mathtt{WeakL}}
              \quad
        \infer
            {\vdots}
            {\Gamma,F^{(jk)}\vdash_j F^{(ij)\circ(ik)},\Delta}
            {\ }
      }
      {\Gamma,F^{(jk)}\der_j\Delta}
      {\mathtt{Cut}}
    $$
where the rightmost premise of the cut rule is given by Corollary \ref{cor:id}, plus
      weakenings, since $(jk)^{-1}(j)=((ij)\circ(ik))^{-1}(j)=k$.

  Concerning $\mathtt{Neg2}$ we have: if $j=i$, then 
  $\mathtt{Neg2}$ is an instance of $\mathtt{NegL}$. Otherwise:
  $$
  \infer
      {
        \infer
        {
        \infer
            {
              \infer
                  { \Gamma^{(ij)} \der_i F, \Delta^{(ij)}\quad i\neq j}
                  {\Gamma^{(ij)}, F^{(ij)} \der_i  \Delta^{(ij)}}
                  {\mathtt{NegL}}
            }
            {\Gamma,F\vdash_j\Delta}
            {\mathtt{Sym}}
        }
        {\Gamma,F,F^{(ik)}\vdash_j\Delta}
         {\mathtt{WeakL}}
        \quad
        \infer
            {\vdots}
            {\Gamma,F^{(ik)}\vdash_j F,\Delta}
            {\ }
      }
      {\Gamma,F^{(ik)}\der_j\Delta}
      {\mathtt{Cut}}
    $$
  where the rightmost premise of the cut rule is given by Corollary \ref{cor:id}, plus
      weakenings, since $(ik)^{-1}(j)=id^{-1}(j)=j$, and $F=F^{id}$.
\end{proof}


An expected property of $n\mathrm{LK}$ is stated in the following lemma.


\begin{lemma}\label{Lemma:identity}
  For all $i\neq j\in \hat n$, $\prov {\e_i} j {}$.
  \end{lemma}
\begin{proof}
$$
\infer
{\infer {\ }{ \der_i\e_i}\ (\mathtt{Const}) \quad\quad i\neq j }
{ \e_i\der_j}\ {(\mathtt{Neg1}) }
$$
\end{proof}

We conclude this section with two examples of derivations in 3LK.

\begin{example}$(\mathrm{The\ exclusion\ of\ the\ } (n+1)\mathrm{th})$
For each formula $F\in\cal F_n$ and dimension $i\in \{1,\ldots,n\}$, the sequent $\vdash_i F^{(1i)},F^{(2i)},\ldots,F^{(ni)}$ is provable. 
Let us see the case $n=3, i=1$ 
$$\infer
{F\der_2 {F^{(12)}}^{(12)}\quad \quad  F\der_3 {F^{(13)}}^{(13)}  }
{\der_1 F, F^{(12)},F^{(13)}} {\mathtt{NegR}}
$$ 
Both premises are provable by Corollary \ref{cor:id}, since
${F^{(12)}}^{(12)}={F^{(13)}}^{(13)}=F$.
\end{example}

\begin{example}
  In $2\mathrm{LK}$ $q(x,y,z)$ reads ``if $x$ then $y$ else $z$''. In general,
  $q(x,y_1,\ldots,y_n)$ reads ``if $x=\e_1$ then $y_1$ else
if $x=\e_2$ then $y_2$ else \ldots else if $x=\e_n$ then $y_n$''.
The following $3LK$ derivation of the sequent
$\der_1 q(\e_1,q(\e_2,F,\e_1,G),H,K)$
illustrates how the first argument
of the $q$ connective acts as a projector. This sequent is provable  since
$q(\e_1,q(\e_2,F,\e_1,G),H,K)$ is equivalent to 
$q(\e_2,F,\e_1,G)$,  which in turn is equivalent to $\e_1$, which is
provable in dimension $1$.

\Small{
$$\infer
{\infer{\infer{\vdots}{\e_1,\e_2\der_1 F} \quad\quad
    \infer{\infer {\ }{\der_2 \e_2}{\mathtt{Const}}}{\e_1^{(12)},\e_2 \der_2 \e_1^{(12)}}{\mathtt{WeakL^+}} \quad\quad \infer{\vdots}{\e_1^{(13)},\e_2\der_3 G^{(1,3)}}   } {e_1\der_1 q(\e_2,F,\e_1,G)}{\mathtt{qR}}\quad\quad  \infer{\vdots}{e_1\der_2 H^{(12)}} \quad \quad \infer{\vdots}{e_1\der_3 K^{(13)}}   }
{\der_1 q(\e_1,q(\e_2,F,\e_1,G),H,K)} {\mathtt{qR}}
$$ 
}
The pending premises are all provable by  Lemma \ref{Lemma:identity}  and
weakenings.

\end{example}

\section{Semantics of  $n\mathrm{LK}$}\label{sec:sem}

{\em Environments}, ranged over by $v,u,v_1,\ldots$ are   functions from the set $V$ of propositional variables to
the set $\hat n$ of dimensions, that should be considered here as generalised truth values.

\begin{definition}
Given an environment $v$, $n\mathrm{LK}$ formulas are interpreted into $\hat n$ as follows:
\begin{itemize}

\item $\sem {X^\pi} v= \pi(v(X))$;
\item $\sem {e_i} v= i$, for all $i\in\hat n$;
\item $\sem {q(F,G_1,\ldots,G_n)} v =  \sem {G_i} v$ if ${\sem F v }=i$.

\end{itemize}
\end{definition}

In this setting the notions of tautology, satisfiability etc. are relative to dimensions. For instance,
$F$ is an {\em $i$-tautology} if $\sem F v =i$,  for all $v$.

The following generalisation of the notion of logical consequence is used in order to prove the soundness of the 
calculus.
\begin{definition}
Given two sequences of formulas $\Gamma,\Delta$, and a dimension $i\in\hat n$ we say that $\Delta$ is an {\em $i$-logical consequence}
of $\Gamma$, written $\Gamma \models_i \Delta$, if for all environments $v$, if $\sem G v =i$ for all $G$ in $\Gamma$,
then there exists $F$ in $\Delta$ such that $\sem F v=i$.
\end{definition}

If $\Gamma \models_i \Delta$, then  we say that the sequent $\Gamma \vdash_i \Delta$ is \emph{valid}.

\begin{lemma}\label{Lemma:sym}
For all $F\in \mathcal F_n$, environment $v$, dimensions $i$ and $j$, we have $\sem F v=i$ if and only if  $\sem {F^{(ij)}} v = j$.
\end{lemma}
\begin{proof}
The proof is an easy induction on $F$. We describe case $F=q(G,H_1,\dots,H_n)$:
$ \sem {q(G,H_1,\dots,H_n)} v= \sem {H_{\sem G v}} v   = i$ iff by induction hypothesis $j= \sem {H^{(ij)}_{\sem G v}} v =  \sem {q(G,H_1^{(ij)},\dots,H_n^{(ij)})} v = \sem {q(G,H_1,\dots,H_n)^{(ij)}} v$.
\end{proof}

\begin{proposition}[Soundness]\label{sound}
For all $i\in\hat n$, $\Gamma$ and $\Delta$, if $\prov \Gamma i \Delta$ then $\Gamma\models_i\Delta$.
\end{proposition}

\begin{proof}
By induction on the proof of $\Gamma \der_i \Delta$.
If the last rule is an instance of $\mathtt{Const}$, 
$\mathtt{Id}$, $\mathtt{WeakL}$, $\mathtt{ WeakR}$,   $\mathtt{ConL}$, $\mathtt{ConR}$
or  $\mathtt{Cut}$, the conclusion is immediate. 
If it is $\mathtt{Sym}$, $\mathtt{NegL}$ or $\mathtt{NegR}$,
the conclusion follows from  Lemma \ref{Lemma:sym}.

Let us consider the case $\mathtt{qL}$, using the same notation as in 
Figure \ref{nLK2}: given $v$ such that 
\begin{itemize}
\item for all $R$ in $\Gamma$, $\sem R v=i$.
\item $\sem {q(F,G_1,\ldots,G_n)} v = \sem {G_{\sem F v}} v =i$,
\end{itemize}
we have to show that there exists $H$ in $\Delta$ such that $\sem H v=i$.
Let $j={\sem F v}$, and consider the $j$-th premise: $\Gamma^{(ij)},F,G_j^{(ij)}\der_j\Delta^{(ij)}$.
By Lemma \ref{Lemma:sym}, for all $R\in\Gamma$, we have $\sem {R^{(ij)}} v=j$, and $\sem {G_j^{(ij)}} v=\sem {(G_{\sem F v})^{(ij)}} v= j$.
Moreover $\sem F v=j$, and by ind. hyp.  $\Gamma^{(ij)},F,G_j^{(ij)}\models_j\Delta^{(ij)}$.
Then there exists $H$ in $\Delta$ such that $\sem {H^{(ij)}} v=j$, 
and we conclude $\sem H v=i$ by Lemma \ref{Lemma:sym}.

The case  $\mathtt{qR}$ is similar, and omitted.
\end{proof}

\section{Completeness}\label{anubi}
In this section  we adapt to $n\mathrm{LK}$ the completeness proof for the classical propositional calculus based on the use of the Lindenbaum algebra of formulas.

\subsection{The Lindenbaum algebra $L_n$}\label{sec:lind}


\begin{definition}
Two formulas $F$ and $G$ of $n\mathrm{LK}$ are {\em equivalent}, written
$F\sim G$, if, for all $i\in\hat n$, both $F\vdash_iG$ and 
$G\vdash_iF$ are provable.
\end{definition}

\begin{lemma}\label{lem:equiv}
The relation $\sim$ is an  equivalence relation on the set 
of formulas of $n\mathrm{LK}$.
\end{lemma}

\begin{proof}
Symmetry is immediate. 
Reflexivity is shown in Corollary \ref{cor:id}.
Transitivity follows  from the cut rule.
\end{proof}

\begin{lemma}\label{lem:congr}
If $F_i\sim G_i$ for $0\leq i\leq n$, then 
$q(F_0,\ldots,F_n)\sim q(G_0,\ldots,G_n)$.
\end{lemma}
\begin{proof}
We can apply $\mathtt{qL}$:

$$
\infer
{\Big\{ F_0,F_j^{(ij)}  \der_j q(G_0,\ldots,G_n)^{(ij)}  \Big\}_{j\in\hat n}}
{q(F_0,\ldots,F_n)\der_i  q(G_0,\ldots,G_n) }\ {(\mathtt{qL})}
$$

Therefore, we have to prove the sequents 
${F_0,F_j^{(ij)}} \vdash_j q(G_0,\ldots,G_n)^{(ij)} $, for $1\leq j\leq n$. Let us consider 
first the case $j=i$: 

$$
\infer
{\ldots\ldots \ \ \infer{F_i\der_i G_i}{F_0,F_i,G_0\der_i G_i}{(\mathtt{Weak^+})} \ \ \ldots\ldots     
\infer{ \infer {F_0\der_i G_0  } 
{F_0^{(ik)},G_0\der_k}{(\mathtt{Neg1})}} 
{F_0^{(ik)},F_i^{(ik)},G_0\der_k G_k ^{(ik)}}  {(\mathtt{Weak^+})} }
{F_0,F_i\der_i q(G_0,\ldots,G_n)} {(\mathtt{qR})}
$$
In the above proof  the $i$-th and $k$-th premisses are developed, for some  $k\not=i$.
All the other premisses are similar to the $k$-th. The conclusion follows.

\noindent Let us consider now the case  $j\not=i$: 
$$
\infer
{\ldots \ \ \infer{\infer{F_j\der_i G_j}{F_j^{(ij)}\der_j G_j^{(ij)}}   
{(\mathtt{Sym})}
}
{F_0,F_j^{(ij)},G_0\der_j G_j^{(ij)}}{(\mathtt{Weak^+})} \ \ \ldots     
\infer{ \infer {F_0\der_j G_0} 
{F_0^{(jk)},G_0\der_k}{(\mathtt{Neg1})}} 
{F_0^{(jk)},{F_j^{(ij)}}^{(jk)},G_0\der_k {G_k ^{(ij)}}^{(jk)}}  {(\mathtt{Weak^+})} }
{F_0,F^{(ij)}_j\der_j q(G_0,\ldots,G_n)^{(ij)}} {(\mathtt{qR})}
$$
In the above proof  the $j$-th and $k$-th premisses are developed, for some  $k\not=j$.
All the other premisses are similar to the $k$-th. The conclusion follows.
\end{proof}

By Lemmas \ref{lem:equiv} and \ref{lem:congr} the relation $\sim$ is a congruence on $\mathcal F_n$.

The next lemma explains the meaning of a decorated formula.

\begin{lemma}\label{Lemma:permutations}
For each $H\in{\mathcal F}_n$ and  $\pi\in S_n$,  $H^\pi\sim q(H,e_{\pi(1)},\ldots,e_{\pi(n)})$.
\end{lemma}
\begin{proof}
Given $i$,  let us begin by proving that $\prov {q(H,e_{\pi(1)},\ldots,e_{\pi(n)})} i {H^\pi}$:

$$
\infer
{\{H, \e^{(ij)}_{\pi(j)} \der_j H^{(ij)\circ\pi}\}_{j\in\hat n}}
{q(H,\e_{\pi(1)},\ldots,\e_{\pi(n)}) \der_i H^\pi} {(\mathtt{qL})}
$$
We reason by case analysis on $j$: if $\pi(j)\not=i$, then $\e^{(ij)}_{\pi(j)}=\e_{(ij)\circ\pi(j)}\not =\e_j$. Since
the sequent $\e^{(ij)}_{\pi(j)} \der_j$ is provable by Lemma \ref{Lemma:identity}, then by applying $\mathtt{WeakR}$ and $\mathtt{WeakL}$ we get $\prov {H, \e^{(ij)}_{\pi(j)}} j {H^{(ij)\circ\pi}}$. If $\pi(j)=i$, then we apply Corollary \ref{cor:id} to prove $H\der_j H^{(ij)\circ\pi}$, because  $((ij)\circ\pi)^{-1}(j)=j=id^{-1}(j)$.
We obtain the conclusion $H, \e_j \der_j H^{(ij)\circ\pi}$  by applying $\mathtt{WeakL}$.

Concerning $\prov {H^\pi} i {q(H,\e_{\pi(1)},\ldots,\e_{\pi(n)})}$ we apply $\mathtt{qR}$:

$$
\infer
{\{ H^{(ij)\circ\pi}, H\der_j  \e^{(ij)}_{\pi(j)}\}_{j\in\hat n}}
{H^\pi \der_i q(H,\e_{\pi(1)},\ldots,\e_{\pi(n)})} {(\mathtt{qR})}
$$
If $\pi(j)=i$, then  $\e^{(ij)}_{\pi(j)}=\e_j$ and we are done by $\mathtt{Const}$ and weakening.
If  $\pi(j)\not=i$, then $((ij)\circ\pi)(j)\not=j$. In this case, renaming  $\sigma$ the permutation $(ij)\circ\pi$ and putting $r=\sigma(j)$, we proceed by applying $\mathtt{Neg1}$ with $\Gamma := H$, $F:= H^\sigma$ and $k:=j$: 
$$
\infer
{\infer{H^{(jr)} \der_{r}  H^\sigma}
{   H^{\sigma}, H\der_j    } {(\mathtt{Neg1})}}
{H^{\sigma}, H\der_j  \e^{(ij)}_{\pi(j) }}{(\mathtt{WeakR})}
$$
Since $\sigma^{-1}(r)=j=(jr)^{-1}(r)$, the sequent $H^{(jr)} \der_{r} H^\sigma$ is provable by Corollary \ref{cor:id}.
\end{proof}


\begin{theorem}
The quotient set ${\mathcal F_n}/\!\!\sim$ is the universe  of an $n\mathrm{BA}$,
denoted by $L_n$, and called the Lindenbaum algebra of  $nLK$.
\end{theorem}

\begin{proof}
  We have to prove that:
  \begin{enumerate}
  \item[(B0)] $q(\e_k,H_1,\dots,H_n)\sim H_k$:
  
$$
\infer
{ \infer{H_k^{(ik)} \der_k  H_k^{(ik)}}{\e_k, H_k^{(ik)}\der_k  H_k^{(ik)}}{(\mathtt{WeakL})}\qquad 
\Big\{\infer{\e_k\der_j}{\e_k, H_j^{(ij)}\der_j  H_k^{(ij)} }{(\mathtt{Weak}^+)}\Big\}_{j\in\hat n\setminus \{k\}} } 
{q(\e_k,H_1,\dots,H_n) \der_{i} H_k}{(\mathtt{qL})}
$$

 $$
\infer
{ \infer{H_k^{(ik)} \der_k  H_k^{(ik)}}{H_k^{(ik)},\e_k\der_k  H_k^{(ik)}}{(\mathtt{WeakL})}\qquad 
\Big\{\infer{\e_k\der_j}{H_k^{(ij)},\e_k\der_j  H_j^{(ij)} }{(\mathtt{Weak}^+)}\Big\}_{j\in\hat n\setminus \{k\}} } 
{H_k \der_{i} q(\e_k,H_1,\dots,H_n) }{(\mathtt{qR})}
$$

 \item[(B1)] $ q(H,\e_1,\dots,\e_n)\sim H$:
$$
\infer
{ \infer{H \der_i  H}{H,\e_i\der_i  H}{(\mathtt{WeakL})}\qquad 
\Big\{\infer{\e_i\der_j}{H, \e_j^{(ij)} \der_j  H^{(ij)} }{(\mathtt{Weak}^+)}\Big\}_{j\in\hat n\setminus \{i\}} } 
{q(H,\e_1,\dots,\e_n) \der_{i} H}{(\mathtt{qL})}
$$
$$
\infer
{ \infer{ \der_i  \e_i}{H,H\der_i  \e_i}{(\mathtt{WeakL})}\qquad 
\Big\{\infer{\infer{\mathrm{Lemma}\ \ref{lem:neg0}}{ H^{(ij)}, H\der_j}}{H^{(ij)}, H \der_j  \e_j^{(ij)} }{(\mathtt{WeakR})}\Big\}_{j\in\hat n\setminus \{i\}} } 
{H \der_{i} q(H,\e_1,\dots,\e_n)}{(\mathtt{qR})}
$$

  \item[(B2)] $q(H,F,\dots,F)\sim F$:
  
$$
\infer
{\infer{F^{(ij)} \der_{j}  F^{(ij)}}
{   H, F^{(ij)}\der_j  F^{(ij)}  } {(\mathtt{WeakL})}}
{q(H,F,\dots,F) \der_{i} F}{(\mathtt{qL})}\qquad\qquad
\infer
{\infer{F^{(ij)} \der_{j}  F^{(ij)}}
{   F^{(ij)}, H\der_j  F^{(ij)}  } {(\mathtt{WeakL})}}
{F \der_{i} q(H,F,\dots,F)}{(\mathtt{qR})}
$$

   \item[(B3)] {\Small $ q(H,q(F^1_0,\ldots,F^1_{n}),\ldots,q(F^n_0,\ldots,F^n_{n}))\sim q(q(H,F^1_0,\ldots,F^n_0),\ldots,q(H,F^1_{n},\ldots,F^n_{n}))$}:

\noindent  We have to show that, for all $H, F^r_k$ ($1\leq r\leq n$ and $0\leq k\leq n$),
the formulas 
$$H_1=q(H,q(F^1_0,\ldots,F^1_{n}),\ldots,q(F^n_0,\ldots,F^n_{n}))$$
and 
$$H_2=q(q(H,F^1_0,\ldots,F^n_0),\ldots,q(H,F^1_{n},\ldots,F^n_{n}))$$ are equivalent.
For all $i$, the proofs of $H_1\vdash_i H_2$ and of  $H_2\vdash_i H_1$ are trees of branching factor $n$ and depth $5$.
Each leaf of those proof trees  is identified by a sequence of $5$ integers $j,k,l,h,m$ between $1$ and $n$, its branch. We are  going to 
construct one of such branches, and argue that the corresponding leaf is always a provable sequent, by case analysis on the sequence of integers associated to it.
Five different exchanges come into play in this derivation, that we rename for typographical reasons: $\pi:=(ij)$, $\rho:=(jk)$, $\sigma:=(kl)$, $\tau:=(lh)$, $\psi:=(hm)$. 
In the following proof the principal formula of each rule application is depicted in bold.

{\small
$$
\infer{\infer{\infer{\infer{\infer
{ H^{\psi\circ\tau\circ\sigma},{F_0^k}^{\psi\circ\tau},{F_{l}^k}^{\psi\circ\tau\circ\sigma\circ\rho\circ\pi},H^\psi,{F_0^h}^{\psi\circ\tau\circ\sigma\circ\rho},H  \vdash_m {F_{j}^m}^{\psi\circ\tau\circ\sigma\circ\rho\circ\pi}}
{H^{\tau\circ\sigma},{F_0^k}^\tau,{F_{l}^k}^{\tau\circ\sigma\circ\rho\circ\pi},H,{F_0^h}^{\tau\circ\sigma\circ\rho}  \vdash_h  \bold{q(H,{F_{j}^1}^{\tau\circ\sigma\circ\rho\circ\pi},\ldots,{F_{j}^n}^{\tau\circ\sigma\circ\rho\circ\pi})}  }  {(\mathtt{qR})}  }
{H^\sigma,F_0^k,{F_{l}^k}^{\sigma\circ\rho\circ\pi},\bold{q(H,{F_0^1}^{\sigma\circ\rho},\ldots,{F_0^n}^{\sigma\circ\rho})} \vdash_l  q(H,{F_{j}^1}^{\sigma\circ\rho\circ\pi},\ldots,{F_{j}^n}^{\sigma\circ\rho\circ\pi})}{(\mathtt{qL})}  }
{H,\bold{q(F_0^k,{F_1^k}^{\rho\circ\pi},\ldots,{F_{n}^k}^{\rho\circ\pi})} ,q(H,{F_0^1}^\rho,\ldots,{F_0^n}^\rho) \vdash_k q(H,{F_{j}^1}^{\rho\circ\pi},\ldots,{F_{j}^n}^{\rho\circ\pi}) }  {(\mathtt{qL})}  }
{\bold{H_1^\pi},q(H,F_0^1,\ldots,F_0^n)\vdash_j q(H,{F^1_{j}}^\pi,\ldots,{F^n_{j}}^\pi) }  {(\mathtt{qL})}  }
{H_1\der_i \bold{H_2}}  {(\mathtt{qR})}
$$
}

\noindent We denote by $S$ the uppermost sequent of the above branch of the proof. We are going to show that $S$  is provable, by case analysis on $j,k,l,h,m$. First of all, if $h\not= m$, then 
$H^{\psi},H\vdash_m$ is provable by Lemma \ref{lem:neg0} and $S$ is provable by weakening. Hence we are left with the case $h=m$;
by applying the rule $(\mathtt{ConL})$, the uppermost sequent $S$ becomes:
$$H^{\tau\circ\sigma},{F_0^k}^{\tau},{F_{l}^k}^{\tau\circ\sigma\circ\rho\circ\pi},H,{F_0^h}^{\tau\circ\sigma\circ\rho}\vdash_h {F_{j}^h}^{\tau\circ\sigma\circ\rho\circ\pi}.$$
We analyse the cases $k\neq h$ and $k=h$, both splitting in two other cases.
\begin{itemize}
\item[($k\neq h$):] The sequent $H^{\tau\circ\sigma},H\vdash_h$ is provable as follows:
		\begin{itemize}
		\item[($l\neq h$):]  Since in this case $(lh)^{-1}(l) =h= (\tau\circ\sigma)^{-1}(l)$, then we have:
		$$
\infer{\infer{\mathrm{Corollary}\ \ref{cor:id}}{H^{(lh)}\vdash_l H^{\tau\circ\sigma}\quad l\neq h} }
{H, H^{\tau\circ\sigma}\vdash_h}{(\mathtt{Neg1})}
$$
		\item[($l=h$):] $H^{\sigma},H\vdash_l$ follows from Lemma \ref{lem:neg0}.  
		\end{itemize}

\item[($k=h$):] In this case $\tau=\sigma$ and  $\sigma\circ\tau= \tau\circ\sigma=id$. The uppermost sequent $S$ becomes:
$$H,{F_0^k}^{\tau},{F_{l}^k}^{\rho\circ\pi},H,{F_0^k}^{\rho}\vdash_k {F_{j}^h}^{\rho\circ\pi}$$
		\begin{itemize}
		\item[($j\neq l$):] ${F_0^k}^{\tau},{F_0^k}^{\rho}\vdash_k$ is provable as follows, recalling that $k=h$:
		$$
\infer{{F_0^k}^{(jk)\circ(lk)},{F_0^k} \vdash_j}
{{F_0^k}^{(lk)},{F_0^k}^{(jk)}\vdash_k}{(\mathtt{Sym})}
$$
Now the proof of  ${F_0^k}^{(jk)\circ(lk)},F_0^k \vdash_j$ is exactly like that of  $H^{\tau\circ\sigma},H\vdash_h$ above, by considering that
$[(jk)\circ(lk)]^{-1}(j)=l\neq j = id^{-1}(j)$.
		\item[($j= l$):] Both side of the  uppermost sequent contain ${F_{l}^k}^{\rho\circ\pi}$, and we are done.
		\end{itemize}
\end{itemize}
\noindent The proof of ${H_2\der_i H_1}$ is similar.  \qedhere
\end{enumerate}
\end{proof}

\subsection{Completeness of  $n\mathrm{LK}$}\label{sec:ideals}
We start by showing that all valid sequents of the form $\Gamma\der_i F$ are
provable, then we generalise to all the valid sequents of the  form $\Gamma\der_i \Delta$.
Let $\Gamma, F,i$ be such that $\Gamma\models_i F$. To begin with, we define a multideal in   the Lindenbaum algebra of $n\mathrm{LK}$ (see Definition \ref{def:fide})  depending on $\Gamma$ and $i$.
Recall that if $\prov \Gamma i F$  and $F\sim G$ then, by a simple application of  rule $ {(\mathtt{Cut})}$, we get
$\prov \Gamma i G$. 

We denote by $[F]$ the equivalence class of the formula $F$ with respect to the equivalence $\sim$.

\begin{definition}
Given $i\in\hat n$ and a  set  of formulas $\Gamma$, we define ${\mathcal I}^{(\Gamma,i)}_j=\{[F] : \prov {\Gamma} i F^{(ij)}\}$, for $j=1,\dots,n$.
\end{definition} 

\begin{lemma}
 ${\mathcal I}^{(\Gamma,i)}$ is a multideal of the Lindenbaum algebra $L_n$ of $n\mathrm{LK}$.
\end{lemma}
\begin{proof} 
We show that  $({\mathcal I}^{(\Gamma,i)}_j)_{j\in\hat n}$ satisfies the conditions of  Definition \ref{def:fide}.
\begin{itemize}
\item $\e_j\in {\mathcal I}^{(\Gamma,i)}_j$ is  trivially true.

\item if $F\in {\mathcal I}^{(\Gamma,i)}_j$ and $G\in{\mathcal I}^{(\Gamma,i)}_k$ then $q(G,H_1,\ldots,H_{k-1},F,H_{k+1},\ldots,H_n)\in {\mathcal I}^{(\Gamma,i)}_j$, for all formulas $H_1,\ldots,H_n$:

$$
\infer{
{\Big\{\Gamma^{(il)},G \der_l (H^{(ij)}_l)^{(il)}\Big\}_{l\in\hat n\setminus \{k\}}} \\  {\Gamma^{(ik)},G\der_k (F^{(ij)})^{(ik)}}  } 
{\Gamma\der_i q(G,H^{(ij)}_1,\ldots,H^{(ij)}_{k-1},F^{(ij)},H^{(ij)}_{k+1},\ldots,H^{(ij)}_n)} {(\mathtt{qR})}
$$
The rightmost antecedent follows from $\Gamma\der_i F^{(ij)}$ by  rules ${(\mathtt{Weak})+(\mathtt{Sym})}$.
Concerning the leftmost antecedents, we show that they are provable for $l\neq k$.

$$
  \infer
          {\infer{\Gamma\der_i G^{(ik)}} {\Gamma^{(il)},G \der_l}
            {(\mathtt{Neg2})}} 
{\Gamma^{(il)},G \der_l (H^{(ij)}_l)^{(il)}} {(\mathtt{WeakR})}
$$

\item if $F_1,\ldots,F_n\in {\mathcal I}^{(\Gamma,i)}_j$ then, for all $G$, $q(G,F_1,\ldots,F_n)\in {\mathcal I}^{(\Gamma,i)}_j$:

$$
\infer
{\Big\{\infer {\Gamma\der_i F^{(ij)}_k}
{\Gamma^{(ik)},G \der_k (F^{(ij)}_k)^{(ik)}}  {(\mathtt{Sym+Weak})}\Big\}_{k\in\hat n} }
{\Gamma\der_i q(G,F^{(ij)}_1,\ldots,F^{(ij)}_n)} {(\mathtt{qR})}
$$
\end{itemize}
\end{proof}

Given a proper ultramultideal $U=(U_1,\dots,U_n)$ on $L_n$, let $v^U$ be the environment  such that $v^U(X)=i$  if $[X^{id}]\in U_i$. The environment $v^U$ is well-defined since $\bigcup_{i=1}^n U_i = L_n$ and $U_i\cap U_j = \emptyset$ for $i\neq j$.

\begin{lemma} \label{Lemma:envU}

For all formulas 
$F$, $\sem F {v^U}=i$ if and only if $[F]\in U_i$.  
\end{lemma}
\begin{proof}

The proof is by induction on $F$. 

($\Leftarrow$) We analyse only the case $F=X^\rho$.
Let  $[X^\rho]=[ q(X^{id},\e_{\rho(1)},\ldots,\e_{\rho(n)})]\in U_i$. Since $U$ is a ultramultideal, then there exists $j$ such that $[X^{id}]\in U_j$. Then, by definition of multideal, $\e_{\rho(j)}\in U_i$. This implies $\rho(j)=i$ since $U$ is proper. We conclude that $\sem {X^\rho} {v^U}=\rho(v^U(x))=\rho(j)=i$.

($\Rightarrow$) Again, we analyse only the case $F=X^\rho$. Let
$\sem {X^\rho} {v^U}=\rho(v^U(X))=i$ and $[X^{id}]\in U_m$, for some $m$, so that
$v^U(X)=m$ and $\rho(m)=i$. Since
$[X^\rho]=[ q(X^{id},\e_{\rho(1)},\ldots,\e_{\rho(n)})]$ and $[X^{id}]\in U_m$, then  $[X^\rho]\in U_{\rho(m)}=U_i$, by definition of multideal.
\end{proof}

\begin{proposition}[Completeness]\label{compl}
If $\Gamma\models_i F$ then $\prov \Gamma i  F$. 
\end{proposition}
\begin{proof}
Let us suppose that  $\Gamma\not\der_i F$, so that 
$[F]\not\in {\mathcal I}^{(\Gamma,i)}_i$ and the multideal
${\mathcal I}^{(\Gamma,i)}$ is proper.  Then by Proposition \ref{prop:ext1} there exists
a proper  ultramultideal $U$ extending  ${\mathcal I}^{(\Gamma,i)}$ 
such 
$[F]\not\in U_i$.
Remark that for all $G\in\Gamma$, $[G]\in{\mathcal I}^{(\Gamma,i)}_i\subseteq U_i$,
by Corollary \ref{cor:id}.
Hence by Lemma \ref{Lemma:envU}  $\sem F {v^U}\not=i$ and, for all $G\in\Gamma$,  $\sem G {v^U}=i$,
proving that $\Gamma\not\models_i F$.
\end{proof}

In the final part of this section we generalise the completeness to sequents with several formulas in the right-hand side.

Let $\Delta \equiv G_1,\dots,G_k$ be a sequence of  $\mathcal F_n$-formulas and $i\in \hat n$. We define a formula $F^i(\Delta)$ by induction on the length $k$ of the sequence. As a matter of notation, $\Delta\setminus G_1$ denotes the sequence $G_2,\dots,G_k$.
If $k=0$, then we define $F^i(\Delta)=\e_j$ with $j\neq i$\footnote{Every $\e_j$ with $j\neq i$ is
suitable. We can choose $j$ minimum such that $j\neq i$, for instance.}. 
 If $k=1$, then $F^i(\Delta)=G_1$. If $k>1$, then 
 $F^i(\Delta)=q(G_1,F^i(\Delta\setminus G_1),\dots,F^i(\Delta\setminus G_1),G_1,F^i(\Delta\setminus G_1),\dots F^i(\Delta\setminus G_1))$,
 where $G_1$ is at position $i+1$.
 
 \begin{lemma}\label{lem:ij}
   If $\Delta$ contains at least one formula and $i,j\in\hat n$, then
   $(F^i(\Delta))^{(ij)}\sim (F^j(\Delta^{(ij)})$.
   \end{lemma}
 \begin{proof}
   By induction on $|\Delta|\geq 1$.
   If $\Delta=G$, then $(F^i(\Delta))^{(ij)}=(F^j(\Delta^{(ij)})=G^{(ij)}$.
   If $\Delta= G,\Delta'$ then
   $$(F^i(\Delta))^{(ij)}= q(G, F^i(\Delta')^{(ij)},\ldots,F^i(\Delta')^{(ij)},
   G^{(ij)},F^i(\Delta')^{(ij)},\ldots,F^i(\Delta')^{(ij)})
   $$ where $G^{(ij)}$ is at position $i+1$, and
 $$ (F^j(\Delta^{(ij)})=q(G^{(ij)}, F^j(\Delta'^{(ij)} ),\ldots, F^j(\Delta'^{(ij)} ),
   G^{(ij)}, F^j(\Delta'^{(ij)}),\ldots, F^j(\Delta'^{(ij)}))
   $$  where $G^{(ij)}$ is at position $j+1$. It is clear that for all envionments $v$,
   $\sem {(F^i(\Delta))^{(ij)}} v = \sem {(F^j(\Delta^{(ij)}))} v  $, so that 
   $(F^i(\Delta))^{(ij)}\models_k (F^j(\Delta^{(ij)})$ for all $k$, and conversely. Hence 
 $(F^i(\Delta))^{(ij)}\sim (F^j(\Delta^{(ij)})$ follows by Proposition \ref{compl}.
 \end{proof}

\begin{lemma}\label{lem:unomolti}
  For all sequences $\Gamma$ and $\Delta$ of  $\mathcal F_n$-formulas, and for all $i\in \hat n$,we have: 
  $$\prov \Gamma i \Delta\ \Leftrightarrow\ \prov \Gamma i {F^i(\Delta)}.$$
\end{lemma}

\begin{proof} ($\Rightarrow$) By Proposition \ref{sound} we have that $\prov \Gamma i \Delta$ implies $\Gamma \models_i \Delta$.
It is an easy exercise to verify that $\Gamma \models_i \Delta$ if and only if $\Gamma \models_i F^i(\Delta)$. Therefore, by Proposition  \ref{compl} we get the conclusion.

($\Leftarrow$) 
Let $k$ be the size of $\Delta$. The case $k=0$ is settled by the following proof,
where the leftmost sequent is provable by hypothesis, and the rightmost one by
Lemma \ref{Lemma:identity}:

$$
\infer {
  \infer{} {\Gamma\vdash_i \e_j} {}
  \quad\quad  \infer{} {\e_j \vdash_i} {}
}
{\Gamma\vdash_i} {(\mathtt{Cut})}
$$
Then, we proceed by induction on $k>0$. The case $k=1$ is trivial since if $\Delta=G$ then
$F^i(\Delta)=G$. Let $\Delta=G,\Delta'$ be a sequence of length $k>1$.  By induction hypothesis we have that
$\prov {\Gamma''} k \Delta''\ \Leftrightarrow\ \prov {\Gamma''} j {F^k(\Delta'')}$, for all $j\in\hat n$ and all sequences $\Gamma'',\Delta''$ with $|\Delta''|< k$. The first step in the proof of $\Gamma\vdash_i G,\Delta'$ is $(\mathtt{NegR})$:

$$
\infer
    {\infer {}
    {\left\{\Gamma^{(ij)}, G\vdash_j\Delta'^{(ij)}\right\}_{j\neq i}}{}
          }
    {\Gamma\vdash_i G,\Delta'}
    {(\mathtt{NegR})}
$$
We are left with the problem of proving $\Gamma^{(ij)}, G\vdash_j\Delta'^{(ij)}$ for all $j\neq i$.
Since by hypothesis $\prov  \Gamma i q(G, F^i(\Delta'),\ldots,F^i(\Delta'),G,F^i(\Delta'),\ldots,F^i(\Delta'))$, we can use the soundness of $n$LK, the
invertibility of the rule $(\mathtt{qR})$ (see Lemma \ref{lem:rev}), and the completeness result
of Proposition \ref{compl}, to get $\prov {\Gamma^{(ij)}, G} j {F^i(\Delta')^{(ij)}}$, for all $j\neq i$. Using   Lemma \ref{lem:ij} we get $\prov {\Gamma^{(ij)}, G} j {F^j(\Delta'^{(ij)})}$. Now, the inductive hypothesis allows us to conclude that $\prov {\Gamma^{(ij)}, G} j {\Delta'^{(ij)}}$, and we are done.
\end{proof}

\begin{theorem}\label{soundcompl}
  For all sequences $\Gamma$ and $\Delta$ of  $\mathcal F_n$-formulas and for all $i\in \hat n$, we have that 
$\Gamma\models_i \Delta $ iff $\prov \Gamma i \Delta$. 
\end{theorem}

\begin{proof}
By Lemma \ref{lem:unomolti} and by Propositions \ref{sound} and \ref{compl}.
  \end{proof}

\subsection{The classical case: semantics}\label{sec:class_cont}

We are now able to prove the equivalence of $2\mathrm{LK}$ and $\mathrm{PC}$, through the translations $(\_)^\circ$ and$(\_)^\bullet$   defined in Section \ref{nefertiti}. 

  Each environment of $2\mathrm{LK}$ may be seen as an environment of $\mathrm{PC}$, and conversely, simply by exchanging the truth values 2 and 0. In the sequel, we will keep this exchange implicit.

  \begin{lemma}\label{lem:pc2pc}
    Let $F$ be a formula of $2\mathrm{LK}$, $P$ a formula of $\mathrm{PC}$ and $v$  an environment.
    Then $\sem F v =1$ iff $\sem {F^\bullet} {v}=1$ and
    $\sem P v =1$ iff $\sem {P^\circ} {v}=1$.
\end{lemma}
\begin{proof}
Both statements are proven by  straightforward  inductions, on $F$ and $P$ respectively.
\end{proof}

  \begin{corollary}\label{try}

    \begin{enumerate}
      \item  \label{cor:eq}

    Let $\Gamma$, $\Delta$ be sequences of $2\mathrm{LK}$ formulas and
     $\Gamma_1$, $\Delta_1$ sequences of $\mathrm{PC}$ formulas. Then  
     $\Gamma \models_1 \Delta$ iff $\Gamma^\bullet \models \Delta^\bullet$, and 
  $\Gamma_1 \models  \Delta_1$  iff $\Gamma_1^\circ \models_1 \Delta_1^\circ$.
\item \label{equiv}
  Let $\Gamma$, $\Delta$ be sequences of $2\mathrm{LK}$ formulas and
     $\Gamma_1$, $\Delta_1$ sequences of $\mathrm{PC}$ formulas. Then  
     $\prov \Gamma 1 \Delta$ iff $\prov {\Gamma^\bullet} {} {\Delta^\bullet}$, and
      $\prov {\Gamma_1}  {} {\Delta_1}$ iff $\prov {\Gamma_1^\circ} 1 {\Delta_1^\circ}$
\end{enumerate}
   \end{corollary}

   \begin{proof}
  (1)     follows from Lemma \ref{lem:pc2pc}.

(2)  Since the calculus of Figure \ref{PC} is sound and complete for the Propositional Calculus, the result follows from Theorem \ref{soundcompl} and Corollary
  \ref{try}(1).
    \end{proof}



\section{Equivalence of $n\mathrm{CL}$ and  $n\mathrm{LK}$}\label{sec:equiv}
Let $\tau$ be the algebraic type of the pure $n$BAs, $\mathbf{T}_\tau(V)$ be the absolutely free term $\tau$-algebra over the countable set $V$ of generators, whose universe is the set of
the $n\mathrm{CL}$ formulas, and $\mathbf n$ be the pure $n$BA of universe 
$\{1,\ldots,n\}$ of Example \ref{exa:n}. More explicitly, a $n\mathrm{CL}$ formula is either one of the constants $\e_1,\ldots,\e_n$, or a variable $X\in V$, or
it is of the form $q(F_0,\ldots,F_n)$ for some a $n\mathrm{CL}$ formulas $F_1,\ldots,F_n$. 
A {\em $\tau$-matrix} is a pair $(\mathbf A, F)$ where $\mathbf A$ is a $\tau$-algebra
and $F\subseteq A$ is a set of {\em designated values}.
A {\em $\tau$-matrix} ${\mathcal A}=(\mathbf A, F)$ defines a logic $L=(\tau,\vdash_{\mathcal A})$ as follows: $\Gamma \vdash_{\mathcal A} \phi$ if for any homomorphism
$h:\mathbf{T}_\tau(V)\rightarrow \mathbf A$, if $h(\psi)\in F$ for all
$\psi\in\Gamma$, then $h(\phi)\in F$.

In order to compare the
logics $n\mathrm{CL}$ defined in \cite{SBLP} and the sequent calculus $n\mathrm{LK}$, we focus
on the logics of the form $(\tau,\vdash_{(\mathbf n,\{i\})})$, for $i=1,\ldots n$.
The homomorphisms from 
$\mathbf{T}_\tau(V)$ to $\mathbf n$ and the environments defined in Section \ref{sec:sem} are in bijective correspondence: given a homomorphism $h:\mathbf{T}_\tau(V)\rightarrow \mathbf n$, let $v_h$ be the environment defined by $v_h(X)=h(X)$, for all $X\in V$;  symmetrically,
given $v$, let $h_v$ be the unique homomorphism extending $v$. 
The map  $v\mapsto v_h=(h\mapsto h_v)^{-1}$ is a bijection, such that $h_{v_h}=h$ and $v_{h_v}=v$.

In order to compare the two systems, 
we need to translate $n\mathrm{CL}$ formulas into $n\mathrm{LK}$ formulas.

\begin{definition}
The  translation $(\_)^\dagger$ from the set of $n\mathrm{LK}$ formulas to the set of $n\mathrm{CL}$ formulas is defined as follows:

\begin{itemize}
\item ${\e_{i}}^\dagger=\e_{i}$;
\item $(X^{\rho})^\dagger=q(X, e_{\rho(1)},..., e_{\rho(n)})$;
\item ${q(F, G_{1}, ...,  G_{n})}^\dagger=q( {F}^\dagger, {G_{1}}^\dagger, ..., {G_{n}}^\dagger)$.
\end{itemize}
\end{definition}

The proof of the following lemma, omitted,  is an easy induction on the size of
$n\mathrm{LK}$ formulas. Notice that the item (2) below is obtained by (1) applied to the environment $v_h$.

\begin{lemma}\label{lem:le1}
  For every $n\mathrm{LK}$  formula $F$,  environment $v$ and  homomorphism $h:\mathbf{T}_\tau(V)\rightarrow\mathbf n$, we have that:
\begin{enumerate}
\item  $\sem F v=h_v(F^\dagger)$, and
\item  $h({F}^\dagger)=\sem F {v_h}$.
\end{enumerate}
  \end{lemma}

\begin{lemma}\label{lem:le2}
  For every $n\mathrm{LK}$ sequent $\Gamma\vdash_i F$, we have:
  $\Gamma\models_i F$ iff $\Gamma^\dagger\vdash_{(\mathbf n,\{i\})} F^\dagger$.
\end{lemma}
\begin{proof}
Immediate by Lemma \ref{lem:le1}.
\end{proof}

\begin{proposition}\label{prop:equiv}
   For every $n\mathrm{LK}$ sequent $\Gamma\vdash_i F$, we have that 
$\Gamma\vdash_i F$ iff $\Gamma^\dagger\vdash_{(\mathbf n,\{i\})} F^\dagger$.
\end{proposition}
\begin{proof}
  By Lemma \ref{lem:le2} and Theorem \ref{soundcompl}.
  \end{proof}

For the sake of completeness, we present below the Hilbert-style system, defined in \cite{SBLP}, axiomatising the logic 
$(\tau,\vdash_{(\mathbf n,\{i\})})$.
By
Proposition \ref{prop:equiv}, this axiomatic system is equivalent to $\vdash_i$, up to the $(\_)^\dagger$ translation of formulas.
We write $\vdash_{\e_i}$ for the derivations in the Hilbert-style system.

Given $\varphi,\psi\in\mathbf{T}_\tau(V)$ and $j,k,l\leq n$, let
\begin{itemize}
\item $\overline{\e_j}[\e_k/l]=\e_j,\e_j,\ldots,\e_k,\ldots,\e_j,\e_j$ ($\e_k$ in $l$-th position).
\item $\varphi\leftrightarrow^{jk}\psi=q(\varphi,q(\psi,\overline{\e_k}[\e_j/1]),\ldots,q(\psi,\overline{\e_k}[\e_j/n]))    $
\end{itemize}

Fixing $j\neq i$, the axioms and rules of the axiomatisation of the logic $(\tau,\vdash_{(\mathbf n,\{i\})})$ are the following (\cite{SBLP}):

\[
\begin{array}[c]{ll}
  A1 & \varphi\leftrightarrow^{i,j}\varphi\\
A2 & q(  \varphi,\e_{1},...,\e_{n})  \leftrightarrow^{i,j}\varphi\\
A3 & q(  \varphi,\overrightarrow{\psi})  \leftrightarrow^{i,j}
\psi\\
A4 &
\begin{array}
[c]{l}
q(  \varphi,q(  \psi_{1},\chi_{11},...,\chi_{1n})
,...,q(  \psi_{n},\chi_{n1},...,\chi_{nn})  )  \\
\leftrightarrow^{i,j}q(  q(  \varphi,\psi_{1},...,\psi_{n})
,...,q(  \varphi,\chi_{1n},...,\chi_{nn})  )
\end{array}
\\
R1 & \varphi,\varphi\leftrightarrow^{i,j}\psi\vdash_{\e_i}\psi\\
R2 & \varphi\leftrightarrow^{i,j}\psi\vdash_{\e_i}\psi\leftrightarrow
^{i,j}\varphi\\
R3 &
\begin{array}
[c]{l}
\varphi_{1}\leftrightarrow^{i,j}\psi_{1},...,\varphi_{n+1}\leftrightarrow
^{i,j}\psi_{n+1}\\
\vdash_{\e_i}q(  \varphi_{1},...,\varphi_{n+1})  \leftrightarrow
^{i,j}q(  \psi_{1},...,\psi_{n+1})
\end{array}
\\
R4 & \varphi\vdash_{\e_i}\varphi\leftrightarrow^{i,j}\e_{i}\\
R5 & \varphi\leftrightarrow^{i,j}\e_{i}\vdash_{\e_i}\varphi
\end{array}
\]



\section{Cut Admissibility}\label{sec:main}

We prove in this section that all the valid sequents admit a cut-free proof, thus obtaining an alternative proof of
completeness,  and a cut admissibility result.
Canonical cut-free proofs are obtained by repeatedly applying the invertible rules  $\mathtt{qR}$ and $\mathtt{qL}$ until we attain sequents that contain no occurrence of the $q$ operator. Subsequently, invertible instances of weakening are applied until we reach sequents that contain only variables and constants. If constants are present, these sequents are easily provable without cuts, while for those containing only variables a slightly more complex analysis is required.
The first observation is that the rules 
$\mathtt{qL}$ and $\mathtt{qR}$ are invertible, meaning that if their conclusion is a valid sequent then all their premises are valid sequents.
\begin{lemma}\label{lem:rev}
The rules $\mathtt{qL}$ and $\mathtt{qR}$ are invertible.
\end{lemma}
\begin{proof}
  Concerning the rule $\mathtt{qL}$, let us suppose that
  $\Gamma,q(F,G_1,\ldots,G_n)\der_i \Delta$ is valid. We have to prove that
  $\Gamma^{(i,j)},F, G_j^{(i,j)}\der_j \Delta^{(i,j)}$ is valid, for all $j$.
  Let us suppose that an environment $v$ is such that $\sem H v=j$ for all
  $H\in \Gamma^{(i,j)}$, $\sem F v=j$ and $\sem {G_j^{(i,j)}} v =j$. Then by Lemma
  \ref{Lemma:sym} $\sem K v=i$ for all
  $K\in \Gamma$, and $\sem {q(F,G_1,\ldots,G_n)} v =i$. From the validity of
  $\Gamma,q(F,G_1,\ldots,G_n)\der_i \Delta$ it follows that there exists $H\in \Delta$
  such that $\sem H v=i$. Hence $\sem {H^{(i,j)}} v=j$, and we are done.

  The case of  $\mathtt{qR}$ is similar, and omitted.
\end{proof}

Hence, all valid sequents admit a cut-free proof iff all valid sequents
containing only constants and decorated variables do. In fact, given a valid sequent $\Gamma\vdash_i \Delta$, one can apply 
 the rules  $\mathtt{qR}$ and $\mathtt{qL}$ as long  as possible, in whatever order.
  By Lemma \ref{lem:rev}, this produces a $n$-branching tree whose leaves are valid sequents containing only constants and decorated variables. 

A second simple observation that allows us  to restrict a bit more the set of sequents to be considered is that some instances of the weakening rules are invertible, namely:
\begin{fact}\label{fact:w}
  Weakenings of the form
  $\infer{\Gamma\der_i\Delta}
  {\Gamma,\e_i\der_i\Delta}
  {\mathtt{WeakL}}$
  and 
$\infer{\Gamma\der_i\Delta}
  {\Gamma\der_i\e_j,\Delta}
  {\mathtt{WeakR}}$, with $j\neq i$,
  are invertible.
\end{fact}

Let us call {\em atomic sequents} those containing only constants and decorated variables, and such that if they are of the form $\Gamma,\e_j\der_i\Delta$
(resp. $\Gamma\der_i\e_j,\Delta$) then $j\neq i$ (resp. $j=i$).

In order to prove that all valid sequents admit a cut-free proof
it is enough to prove that all valid atomic sequents do.
The following easy lemma allows us to get rid of constants.
\begin{lemma}\label{lem:atomic}
 All atomic sequents containing at least a constant are valid and admit a cut-free proof.
\end{lemma}
\begin{proof}
  The validity is immediate. Concerning the provability, sequents of the form
  $\Gamma\der_i \e_i,\Delta$ are proved by using the rule $\mathtt{Const}$ followed by weakenings, and those of the form $\Gamma,\e_j\der_i\Delta$ where $i\neq j$ are proved by using the rule $\mathtt{Const}$ followed by $\mathtt{NegL}$ and weakenings.
  \end{proof}

We are left with the problem of showing that, if a sequent containing only decorated variables 
is valid, then it admits a cut-free proof.  Any such sequent is of the form:
$$
X_1^{\rho^1_1},\ldots,X_1^{\rho^1_{k_1}},\ldots X_s^{\rho^s_1},\ldots,X_s^{\rho^s_{k_s}}\der_i
X_1^{\sigma^1_1},\ldots,X_1^{\sigma^1_{l_1}},\ldots X_s^{\sigma^s_1},\ldots,X_s^{\sigma^s_{l_s}}
$$
where, for all $1\leq r\leq s$, $k_r+l_r>0$. This means that each variable $X_r$ may appear only on the left hand side of the sequent, or only on the rigth hand side, or in both.

The following lemma provides a characterisation of the valid sequents containing only decorated variables.

\begin{lemma}\label{lem:valid}
  A sequent of the form 
$$
X_1^{\rho^1_1},\ldots,X_1^{\rho^1_{k_1}},\ldots X_s^{\rho^s_1},\ldots,X_s^{\rho^s_{k_s}}\der_i
X_1^{\sigma^1_1},\ldots,X_1^{\sigma^1_{l_1}},\ldots X_s^{\sigma^s_1},\ldots,X_s^{\sigma^s_{l_s}}
$$ is valid iff
$$ \exists\ r\leq s \ \forall j\in\hat n\  
  [\ (\forall m\leq k_r \ \rho^r_m(j)= i) \ \Rightarrow
    \ (\exists m\leq l_r \ \sigma^r_m(j)=i)\ ].
  $$


\end{lemma}
\begin{proof}
  It is immediate to prove the contrapositive: the given sequent is not valid iff there exists an environment assigning $i$ to all the decorated variables in the left-hand side,  and to none of those in the right-hand side. Such an environment $v$  exists iff
  for every $r\in\{1,\ldots,s\}$ there exists $v(X_r)\in\{1,\ldots n\}$ such that $\rho^r_m(v(X_r))=i$ for all
  $m\in\{1,\ldots,k_r\}$ and $\sigma^r_m(v(X_r))\neq i$ for all
  $m\in\{1,\ldots,l_r\}$. 
\end{proof}

Let us call {\em pure} a sequent of the form 
$X^{\rho_1},\ldots,X^{\rho_k}\der_iX^{\sigma_1},\ldots,X^{\sigma_l}$.

Lemma \ref{lem:valid} says that a sequent containing only decorated variables as above  is valid iff there exists $r\in\{1,\ldots,s\}$ such that
$X_r^{\rho^r_1},\ldots,X_r^{\rho^r_{k_r}}\der_iX_r^{\sigma^r_1},\ldots,X_r^{\sigma^r_{l_r}}$ is a valid pure sequent.
If all the valid pure sequents admit a cut-free proof, then also all the valid sequents containing only decorated variables admit a cut-free proof, obtained by adding the suitable weakenings to the proof of the pure valid sequent they contain.
Hence, we are left with the following bit: 

\begin{lemma}\label{lem:pure}
If a pure sequent is valid then it admits a cut-free proof.
\end{lemma}
\begin{proof}
  Let $S=X^{\rho_1},\ldots,X^{\rho_k}\der_iX^{\sigma_1},\ldots,X^{\sigma_l}$ be a valid pure sequent. For the sake of this proof, let us abbreviate ``$\exists 1\leq m\leq k \ \rho_m(j)\neq i$'' by $A(j)$ and ``$\exists 1\leq m\leq l \ \sigma_m(j)=i$'' by $B(j)$. 
 By Lemma \ref{lem:valid} we know that
  $\forall 1\leq j\leq n\  A(j)\vee B(j)$.
  
  We consider two cases:
  \begin{enumerate}
  \item $A(j)$ holds for at least one  $j\in\hat n$. Then  there exists  $m\in\{1,\ldots,k\}$ such that $\rho_m(j)\neq i$. Let $h=\rho_m^{-1}(i)\neq j$. There are two cases:
    \begin{enumerate}
    \item $A(h)$ holds. Then there exists  $m'\in\{1,\ldots,k\}$, such that $\rho_{m'}(h)\neq i$. So $\rho_{m'}^{-1}(i)\neq h= \rho_{m}^{-1}(i)$. In this case the sequent
      $X^{\rho_m},X^{\rho_{m'}} \vdash_i$ is cut-free provable  as we have shown in part (i) of  Lemma \ref{lem:neg0}. Then we conclude by
      weakening. Hence, $S$ admits a cut-free proof.
    \item  $B(h)$ holds. Then there exists  $m'\in\{1,\ldots,l\}$, such that
      $\sigma_{m'}(h)= i$. So $\sigma_{m'}^{-1}(i)=\rho_{m}^{-1}(i)$. Then
      $X^{\rho_m} \der_i X^{\sigma_{m'}}$ is the conclusion of an instance of $\mathtt{Id}$, so that
$\prov {X^{\rho_m}} i {X^{\sigma_{m'}}}$, and we conclude by
      weakening.
    \end{enumerate}
  \item $B(j)$ holds for all $j\in\hat n$.
Let $f:\hat n\rightarrow\{1,\ldots,l\}$ be such that 
$\sigma_{f(j)}(j)=i$, for $j\in\hat n$.
If we show that $\vdash_i {X^{\sigma_{f(1)}},\ldots,X^{\sigma_{f(n)}}}$ is cut-free provable we are done,
since the provability  of the whole pure sequent follows by weakening. Now consider

    $$
\infer
    {
      \infer
          {}
          {\Big\{X^{\sigma_{f(i)}}\der_j  X^{(ij)\circ\sigma_{f(1)}},\ldots,X^{(ij)\circ\sigma_{f(i-1)}},
            X^{(ij)\circ\sigma_{f(1+1)}},\ldots,X^{(ij)\circ\sigma_{f(n)}}\Big\}_{j\neq i}
          }
          {}
    }
    {\der_i X^{\sigma_{f(1)}},\ldots,X^{\sigma_{f(n)}}}
    {\mathtt{NegR}}
    $$
  \end{enumerate}
  Each of the $n-1$ premises of this instance of $\mathtt{NegR}$ is cut-free provable: let $j\neq i$ and $r\in\hat n$ be the unique truth value such that $\sigma_{f(i)}(r)=j$. Then the sequent  $X^{\sigma_{f(i)}} \der_j X^{(ij)\circ\sigma_{f(r)}}$ is the conclusion of an instance of $\mathtt{Id}$, because $ \sigma_{f(i)}^{-1}(j)  =r=((ij)\circ\sigma_{f(r)})^{-1}(j)$. The whole premise is obtained by weakening.
  
\end{proof}

Summing up:

\begin{theorem}
All valid sequents of the $n$-dimensional propositional calculus admit a cut-free proof.
\end{theorem}
\begin{proof}
  Given a valid sequent $\Gamma\vdash_i \Delta$, apply 
  $\mathtt{qR}$ and $\mathtt{qL}$ as long  as possible, in whatever order.
  By Lemma \ref{lem:rev}, this produces a $n$-branching tree whose leaves are valid sequents containing only constants and decorated variables. Let $\Gamma'\vdash_j \Delta'$
  be one of these leaves. Apply the invertible instances of $\mathtt{WeakL}$ and $\mathtt{WeakR}$ given in Fact \ref{fact:w} as long  as possible, in whatever order. The result is a valid atomic sequent $\Gamma''\vdash_j\Delta''$. If it contains at least a constant, then it admits a cut-free proof  by
Lemma \ref{lem:atomic}, otherwise it is a valid sequent containing only decorated variables. In this last case, by Lemma \ref{lem:valid},  $\Gamma''\vdash_j\Delta''$ may be obtained by suitably weakening a valid pure sequent $\Gamma'''\vdash_j\Delta'''$. By Lemma \ref{lem:pure}
$\Gamma'''\vdash_j\Delta'''$ admits a cut free proof.
The tree rooted in $\Gamma\vdash_i \Delta$ whose definition is sketched above is a cut-free $n\mathrm{LK}$ proof, and we are done.
\end{proof}


\section{Conclusion}

The sequent calculus introduced in this paper exemplifies a new proof-theoretic format for finite-valued logics with $n$ truth values, where (i) every truth-value $e_i$ has its own turnstile $\vdash_i$, and (ii) propositional variables are decorated with elements of the symmetric group ${\cal S}_n$. The former feature generalises what happens in hybrid deduction-refutation systems: in the $2$-valued case, indeed, proofs of $\vdash_2 F$ should be regarded as {\em refutations} of $F$. The latter feature, on the other hand, allows us to identify a formula $F$ and its double-negation not only as regards their semantic interpretation, but even as syntactic objects. Our calculus is sound and complete with respect to the canonical notion of semantic validity of $n$-dimensional propositional formulas. We proved that all valid sequent admit a cut-free proof, as a corollary to our syntactic proof of completeness. 
 This result falls short of a full-blooded cut elimination theorem, in that no algorithm is provided to transform proofs containing cuts into cut-free proofs. The main difficulty in obtaining it is due to the fact that the propositional variables are decorated. A constructive proof of cut elimination remains as a valuable goal for further investigations.\\
\indent 


\subsection*{Acknowledgements}
A. Ledda and F. Paoli were supported by PLEXUS, (Grant Agreement no 101086295) a Marie Sklodowska-Curie action funded by the EU under the Horizon Europe Research and Innovation Programme. They  also express their gratitude for the support of the Italian Ministry of Education, University and Research through the PRIN 2022 project "Developing Kleene Logics and their Applications" (DeKLA), project code: 2022SM4XC8."

\bibliographystyle{plain}
\bibliography{main}

\end{document}